\documentclass[a4paper,onecolumn,accepted=2021-11-25]{quantumarticle}
\pdfoutput=1

\usepackage{amsmath, amssymb,amsthm}
\usepackage{xspace}
\usepackage{dsfont}
\usepackage[colorlinks=true,urlcolor=webblue,linkcolor=webgreen,filecolor=webblue,citecolor=webgreen,pdfpagemode=UseOutlines,pdfstartview=FitH,pdfpagelayout=OneColumn,bookmarks=true]{hyperref}
\usepackage{enumerate}
\usepackage{graphicx}
\usepackage{algorithm, algpseudocode, float}
\usepackage{tikz}
\usepackage{bm}
\usepackage{doi}
\usepackage[normalem]{ulem} 

\definecolor{webgreen}{rgb}{0,.5,0}
\definecolor{webblue}{rgb}{0,0,.5}

\newtheorem{theorem}{Theorem}
\newtheorem{experiment}[theorem]{Experiment}

\newtheorem{construction}[theorem]{Construction}
\newtheorem{condition}[theorem]{Condition}
\newtheorem{defn}[theorem]{Definition}
\newtheorem{lem}[theorem]{Lemma}
\newtheorem{prop}[theorem]{Proposition}
\newtheorem{cor}[theorem]{Corollary}

\numberwithin{theorem}{section}

\newcommand{\NN}{\mathbb{N}}

\newcommand{\id}{\mathrm{id}}

\newcommand{\ket}[1]{| #1 \rangle}
\newcommand{\bra}[1]{\langle #1 |}

\newcommand{\ketbra}[2]{\left|#1\right\rangle\!\!\left\langle #2\right|}

\newcommand{\tr}{\mathrm{Tr}}

\newcommand{\proj}[1]{\ensuremath{|#1\rangle \langle #1|}}

\renewcommand{\rho}{\varrho}

\newcommand{\Hi}{\mathcal{H}}

\newcommand{\hi}{\Hi}

\newcommand{\one}{\mathds 1}

\newcommand{\secpar}{\ensuremath{n}\xspace}
\newcommand{\secparam}{\ensuremath{1^{\secpar}}\xspace}

\newcommand{\nket}[1]{\ket{#1}}
\newcommand{\nbra}[1]{\bra{\! #1}}

{\begin{framed}\begin{small}}
{\end{small}\end{framed}}

\newcommand{\Hyb}{\ensuremath{\textsf{Hyb}}\xspace}
\newcommand{\PiKEM}{\ensuremath{{\Pi^{\textsf{Hyb}}}}\xspace}
\newcommand{\PiKEMS}{\ensuremath{{\Pi_\SKQE^{\textsf{Hyb}}}}\xspace}
\newcommand{\PiKEMP}{\ensuremath{{\Pi_\PKQE^{\textsf{Hyb}}}}\xspace}
\newcommand{\PiHyb}{\ensuremath{{\Pi_{\QSC}^{\textsf{Hyb}}}}\xspace}
\newcommand{\PiMod}{\ensuremath{{\Pi^{\textsf{Mod}}}}\xspace}

\newcommand{\CQHyb}{\PiHyb}
\newcommand{\CQHybGA}{\PiKEM}

\newcommand{\PKE}{\ensuremath{\textsf{PKE}}\xspace}
\newcommand{\DS}{\ensuremath{\textsf{DS}}\xspace}
\newcommand{\PKQE}{\ensuremath{\textsf{PKQE}}\xspace}
\newcommand{\SKE}{\ensuremath{\textsf{SKE}}\xspace}
\newcommand{\SKQE}{\ensuremath{\textsf{SKQE}}\xspace}
\newcommand{\QS}{\ensuremath{\textsf{QS}}\xspace}
\newcommand{\SC}{\ensuremath{\textsf{SC}}\xspace}
\newcommand{\QSC}{\ensuremath{\textsf{QSC}}\xspace}
\newcommand{\muQSC}{\ensuremath{\textsf{muQSC}}\xspace}

\newcommand{\SKES}{\ensuremath{\textsf{SKES}}\xspace}

\newcommand{\expref}[2]{\texorpdfstring{\hyperref[#2]{#1~\ref{#2}}}{#1~\ref{#2}}}

\newcommand{\A}{\ensuremath{\mathcal{A}}\xspace}
\newcommand{\B}{\ensuremath{\mathcal{B}}\xspace}
\newcommand{\C}{\ensuremath{\mathcal{C}}\xspace}

\renewcommand{\H}{\ensuremath{\mathcal H}\xspace}
\newcommand{\I}{\ensuremath{\mathcal{I}}\xspace}
\newcommand{\K}{\ensuremath{\mathcal{K}}\xspace}

\renewcommand{\P}{\ensuremath{\mathcal{P}}\xspace}

\newcommand{\calS}{\ensuremath{\mathcal{S}}\xspace}

\newcommand{\V}{\ensuremath{\mathcal{V}}\xspace}

\newcommand{\states}{\ensuremath{\mathfrak D}\xspace}
\newcommand{\negl}{\ensuremath{\operatorname{negl}}\xspace}
\newcommand{\supp}{\textbf{supp\,}}
\newcommand{\egoketbra}[1]{\ketbra{#1}{#1}}			
\newcommand{\from}{\ensuremath{\leftarrow}}
\newcommand{\bit}{\{0,1\}}
\newcommand{\rand}{\raisebox{-1pt}{\ensuremath{\,\xleftarrow{\raisebox{-1pt}{$\scriptscriptstyle\$$}}\,}}}
\newcommand{\bin}{\left\{0,1\right\}}
\newcommand{\game}[1]{{\ensuremath{\mathsf{Game} \ \textsf{#1}}}\xspace}
\renewcommand{\Game}[1]{{\ensuremath{{\bm{\mathsf{Game}}} \ \textbf{#1}}}\xspace}

\newcommand{\KeyGen}{\ensuremath{\mathsf{KeyGen}}\xspace}
\newcommand{\Enc}{\ensuremath{\mathsf{Enc}}\xspace}
\newcommand{\Sign}{\ensuremath{\mathsf{Sign}}\xspace}
\newcommand{\Dec}{\ensuremath{\mathsf{Dec}}\xspace}

\newcommand{\Ver}{\ensuremath{\mathsf{Ver}}\xspace}

\newcommand{\poly}{\operatorname{poly}}
\newcommand{\algo}{\mathcal}
\newcommand{\inrand}{\raisebox{-1pt}{\ensuremath{\,\xleftarrow{\raisebox{-1pt}{$\scriptscriptstyle\$$}}\,}}}

\newcommand{\adver}{\A}
\newcommand{\badver}{\B}
\newcommand{\chall}{\C}

\newcommand{\simu}{\ensuremath{\algo S}\xspace}

\newcommand{\QAEreal}{\ensuremath{\mathsf{QAE\mbox{-}Real}}\xspace}
\newcommand{\QAEideal}{\ensuremath{\mathsf{QAE\mbox{-}Ideal}}\xspace}
\newcommand{\QINDCCAAtest}{\ensuremath{\mathsf{QCCA2\mbox{-}Test}}\xspace}
\newcommand{\QINDCCAAfake}{\ensuremath{\mathsf{QCCA2\mbox{-}Fake}}\xspace}
\newcommand{\QINDwCCAAfake}{\ensuremath{\mathsf{QwCCA2\mbox{-}Fake}}\xspace}

\newcommand{\OutReal}{\ensuremath{\mathsf{Out\mbox{-}Real}}\xspace}
\newcommand{\OutIdeal}{\ensuremath{\mathsf{Out\mbox{-}Ideal}}\xspace}
\newcommand{\MOutReal}{\ensuremath{\mathsf{M\mbox{-}Out\mbox{-}Real}}\xspace}
\newcommand{\MOutIdeal}{\ensuremath{\mathsf{M\mbox{-}Out\mbox{-}Ideal}}\xspace}

\newcommand{\LWE}{\ensuremath{\mathsf{LWE}}\xspace}

\newcommand{\IND}{\ensuremath{\mathsf{IND}}\xspace}
\newcommand{\CPA}{\ensuremath{\mathsf{CPA}}\xspace}
\newcommand{\CCA}{\ensuremath{\mathsf{CCA1}}\xspace}
\newcommand{\CCAA}{\ensuremath{\mathsf{CCA2}}\xspace}
\newcommand{\CCAAA}{\ensuremath{\mathsf{CCA}}\xspace}
\newcommand{\INDCPA}{\ensuremath{\mathsf{IND\mbox{-}CPA}}\xspace}
\newcommand{\INDCCA}{\ensuremath{\mathsf{IND\mbox{-}CCA1}}\xspace}
\newcommand{\INDCCAA}{\ensuremath{\mathsf{IND\mbox{-}CCA2}}\xspace}

\newcommand{\QIND}{\ensuremath{\mathsf{QIND}}\xspace}
\newcommand{\QINDCPA}{\ensuremath{\mathsf{QIND\mbox{-}CPA}}\xspace}
\newcommand{\QINDCCA}{\ensuremath{\mathsf{QIND\mbox{-}CCA1}}\xspace}
\newcommand{\QINDCCAA}{\ensuremath{\mathsf{QIND\mbox{-}CCA2}}\xspace}
\newcommand{\QINDwCCAA}{\ensuremath{\mathsf{QIND\mbox{-}wCCA2}}\xspace}

\newcommand{\CAE}{\ensuremath{\mathsf{AE}}\xspace}
\newcommand{\QAE}{\ensuremath{\mathsf{QAE}}\xspace}
\newcommand{\QCA}{\ensuremath{\mathsf{QCA}}\xspace}
\newcommand{\cQCA}{\ensuremath{\mathsf{cQCA}}\xspace}
\newcommand{\CAEreal}{\ensuremath{\mathsf{AE\mbox{-}Real}}\xspace}
\newcommand{\CAEideal}{\ensuremath{\mathsf{AE\mbox{-}Ideal}}\xspace}
\newcommand{\wQAEreal}{\ensuremath{\mathsf{wQAE\mbox{-}Real}}\xspace}
\newcommand{\wQAEideal}{\ensuremath{\mathsf{wQAE\mbox{-}Ideal}}\xspace}

\newcommand{\wCAEideal}{\ensuremath{\mathsf{wAE\mbox{-}Ideal}}\xspace}
\newcommand{\wQAE}{\ensuremath{\mathsf{wQAE}}\xspace}
\newcommand{\wCAE}{\ensuremath{\mathsf{wAE}}\xspace}

\newcommand{\DNS}{\ensuremath{\mathsf{DNS}}\xspace}
\newcommand{\cDNS}{\ensuremath{\mathsf{cDNS}}\xspace}

\newcommand{\acc}{\ensuremath{\mathsf{acc}}\xspace}
\newcommand{\crej}{\ensuremath{\mathsf{rej}}\xspace}
\newcommand{\rej}{\ensuremath{\mathsf{rej}}\xspace}

\newcommand{\win}{\ensuremath{\mathsf{win}}\xspace}
\newcommand{\cheat}{\ensuremath{\mathsf{cheat}}\xspace}
\newcommand{\real}{\ensuremath{\mathsf{real}}\xspace}

\newcommand{\SDK}{\ensuremath{\K_{\mathsf{SD}}}\xspace}
\newcommand{\VEK}{\ensuremath{\K_{\mathsf{VE}}}\xspace}
\newcommand{\DK}{\ensuremath{\K_{\mathsf{D}}}\xspace}
\newcommand{\EK}{\ensuremath{\K_{\mathsf{E}}}\xspace}
\newcommand{\SK}{\ensuremath{\K_{\mathsf{S}}}\xspace}
\newcommand{\VK}{\ensuremath{\K_{\mathsf{V}}}\xspace}
\newcommand{\SigEnc}{\ensuremath{\mathsf{SigEnc}}\xspace}
\newcommand{\VerDec}{\ensuremath{\mathsf{VerDec}}\xspace}
\newcommand{\sk}{\ensuremath{\mathsf{sk}}\xspace}
\newcommand{\vk}{\ensuremath{\mathsf{vk}}\xspace}
\newcommand{\ek}{\ensuremath{\mathsf{ek}}\xspace}
\newcommand{\dk}{\ensuremath{\mathsf{dk}}\xspace}

\newcommand{\sdk}{\ensuremath{\mathsf{sdk}}\xspace}
\newcommand{\vek}{\ensuremath{\mathsf{vek}}\xspace}
\newcommand{\ID}{\ensuremath{\mathsf{ID}}\xspace}

\newcommand{\Qu}{\ensuremath{\mathsf{Qu}}\xspace}
\newcommand{\Cl}{\ensuremath{\mathsf{Cl}}\xspace}

\title{Can You Sign A Quantum State?}

\author{
Gorjan Alagic}
\affiliation{
	QuICS, University of Maryland, College Park, MD, USA}
\affiliation{
	National Institute of Standards and Technology, Gaithersburg, MD, USA }
\email{galagic@umd.edu}
 \author{Tommaso Gagliardoni}
 \affiliation{
Kudelski Security, Zurich, Switzerland }
 \email{name.surname@kudelskisecurity.com}
 \author{Christian Majenz}
\affiliation{
Centrum Wiskunde \& Informatica and QuSoft, Amsterdam, Netherlands}
\email{christian.majenz@cwi.nl}

\begin{document}
%
\begin{abstract}
Cryptography with quantum states exhibits a number of surprising and counter-intuitive features. 
In a 2002 work, Barnum et al. argued that these features imply that digital signatures for quantum states are impossible~\cite{BCG+02}. 
In this work, we ask: can all forms of signing quantum data, even in a possibly weak sense, be completely ruled out? We give two results which shed significant light on this basic question. 

First, we prove an impossibility result for digital signatures for quantum data, which extends the result of \cite{BCG+02}. Specifically, we show that no nontrivial combination of correctness and security requirements can be fulfilled, beyond what is achievable simply by measuring the quantum message and then signing the outcome. In other words, \emph{only classical signature schemes exist.}

We then show a positive result: a quantum state can be signed with the same security guarantees as classically, provided that it is also encrypted with the public key of the intended recipient. Following classical nomenclature, we call this notion \emph{quantum signcryption}. Classically, signcryption is only interesting if it provides superior performance to encrypt-then-sign. Quantumly, it is far more interesting: it is the only signing method available. We develop ``as-strong-as-classical'' security definitions for quantum signcryption and give secure constructions based on post-quantum public-key primitives. Along the way, we show that a natural hybrid method of combining classical and quantum schemes can be used to ``upgrade'' a secure classical scheme to the fully-quantum setting, in a wide range of cryptographic settings including signcryption, authenticated encryption, and CCA security.
\end{abstract}

\tableofcontents

\section{Introduction}


The Internet of the future will plausibly include both large-scale quantum computers and high-capacity quantum channels. How will we securely transmit (quantum) data over the resulting ``quantum Internet''? Methods based on entanglement (e.g., teleportation) are costly, using many rounds of interaction to build a shared state which must be at least as large as the data itself. Classically, encryption and authentication offer a non-interactive approach with several attractive features: (i.) keys exchanged over public channels, (ii.) a short key suffices for transmitting unlimited amounts of data, and (iii.) security guarantees are maximal for both secrecy and authenticity. Can we encrypt, authenticate, and sign quantum data to the same standard? 

Classical digital signatures, for instance, are ubiquitous in everyday cryptography, with applications ranging from secure software distribution and email signatures to e-governance and cryptocurrencies. Given their importance in the classical world, it is natural to ask whether it is possible to devise digital signature schemes for quantum data. Unlike its private-key analogue (see, e.g., \cite{DNS12,HLM16,Portmann17,Dulek2018} and references therein),  this question has been considered in only one previous work~\cite{BCG+02}. There, the authors argue in a brief discussion that 
quantum digital signatures must be impossible. They also suggest that one can use classical public-key cryptography and one-time quantum authentication to build a scheme they call ``public-key quantum authentication.'' However, despite having many clever insights, this work entailed no formal security definitions or rigorous proofs regarding quantum signatures, and the theory remains mostly undeveloped. 

In this work, we return to the problem of signing quantum states, with a rigorous and formal approach. We prove that any quantum digital signature scheme that is correct must be insecure. In addition, we fully address the question of whether there exists any feasible trade-off in security and correctness requirements which makes signing quantum states possible.

%
%

\subsection{Summary of Results}


\subsubsection{Impossibility of quantum signatures}

We first define quantum signature schemes, leaving as much room as possible for achievable notions. 

\begin{defn}[informal] 
A {\em \textbf{quantum signature scheme}} (or $\QS$) is a pair $(\Sign, \Ver)$ of keyed quantum polynomial time (QPT) algorithms. 
\begin{itemize}
\item A \QS is \textbf{correct for a map} $N$ if $\|N \circ \Ver_{\vk} \circ \Sign_{\sk} - N\|_\diamond \leq \negl(n)$. 
\item A \QS is simply \textbf{correct} if it is correct for $N$ being the identity map.
\item A \QS is $\varepsilon$-\textbf{one-time-secure} for a map $L$ if, for every adversary $\A$, the ``attack map'' $L \circ \Ver_{\vk} \circ \mathcal A \circ \Sign_{\sk}$ can be simulated (up to $\varepsilon$ loss) by an algorithm which either applies $L$ or discards the input.
\end{itemize}
\end{defn}
Requiring correctness only for certain maps $N$, or security only for certain maps $L$, weakens the definition in that the (secure) further use of the verified state is limited. We have relaxed the classical requirement that a signature is additional information that accompanies the message; quantumly, this is impossible due to no-cloning. 
Instead, correctness requires that (some part of) the message is recovered at verification time. 
See \expref{Section}{sec:defQSC} for a discussion.

We prove two strong impossibility results. First, we show that full correctness implies negligible security, even in a very weak one-time sense which only guarantees the integrity of the outcomes of a single two-outcome measurement. Second, if we fix a pair of measurements and then ask only for correctness and security of their outcomes, then those measurements must commute.

\begin{theorem}[informal]\label{thm:impossibility-summary}
	Let $\Pi$ be a quantum signature scheme.
	\begin{enumerate}
		\item If $\Pi$ is correct, then it is at most $(1-\negl(\secpar))$-one-time secure for any two-outcome measurement $M$.
		\item If $\Pi$ is correct and $\varepsilon$-one-time-secure for a pair $\{M_0,M_1\}$ of two-outcome measurements, then $M_0$ and $M_1$ are indistinguishable from $(1-\varepsilon)$-commuting. 
	\end{enumerate}
\end{theorem}

The measurement maps above are viewed as channels which measure and then output the outcome, discarding the post-measurement state. This result shows convincingly that signature schemes can only sign classical information. \expref{Theorem}{thm:impossibility-summary} is proven via a reduction that constructs a malleability attack from a distinguishing attack. The key-dependence of the latter prevents an application of the result from \cite{BCG+02} and necessitates more refined techniques based on the Stinespring dilation as well as a related approximation result \cite{Kretschmann2008}.

\subsubsection{Quantum signcryption}

Our impossibility results appear to be devastating to prospects for public-key cryptography with quantum data. They seem to imply that authenticated communication requires each pair of parties to share a secret key. For networks with a large number of parties, or where parties frequently come and go, this is an unwieldy and highly inefficient solution. 

In the second part of our work, we show that this particular application of signatures can still be achieved.  The key observation is that impossibility can be circumvented \emph{if we also encrypt the message}. This rules out public verification (in that case, the adversary can decrypt, making the encryption pointless.) Instead, for each transmission, we select an intended recipient and encrypt using their public key. Classically, such combined schemes are called {\em signcryption}, and are of interest only insofar as they provide efficiency gains over combined encryption and signing~\cite{Zheng97}. This is in stark contrast to the quantum world: our results show that signcryption provides \emph{the only way} to achieve integrity and authenticity without a pre-shared secret.


\begin{defn}[informal]
A {\em \textbf{quantum signcryption scheme}} (or $\QSC$) is a triple of QPT algorithms:
\begin{enumerate}
\item (key generation) $\KeyGen(1^n):$ output $(\sdk, \vek) \from \{0,1\}^{\poly(n)}$.
\item (signcrypt) $\SigEnc_{\sdk, \vek}: \states(\H_M) \rightarrow \states(\H_C)$
\item (verified decrypt) $\VerDec_{\vek, \sdk}: \states(\H_C) \rightarrow \states(\H_M \oplus \ket{\bot})$
\end{enumerate}  
such that $\| \VerDec_{\vek_S, \sdk_R} \circ \SigEnc_{\sdk_S, \vek_R} - \id_M \oplus 0_\bot \|_\diamond \leq \negl(n),$ where $\id$ denotes the identity.
\end{defn}
In the envisioned usage, each party on a network first runs key generation, publishing their ``verify/encrypt key'' $\vek$ while keeping their ``sign/decrypt key'' $\sdk$ private. When a sender $S$ wishes to send a state $\sigma_M$ to a recipient $R$, they apply $\SigEnc_{\sdk_S, \vek_R}$ using their private key and the public key of $R$. When $R$ receives the signcryption $\rho_C$, they apply $\VerDec_{\vek_S, \sdk_R}$ using their private key and the public key of $S$. Note that ``$\oplus \ 0_\bot$'' above indicates that this process of signcryption followed immediately by verified decryption (with the correct keys) never rejects. We remark that the general definition for the multi-user setting will also need to keep track of ``user IDs'' in order to prevent identity fraud.

\paragraph{Quantum signcryption security.}

Signcryption security is naturally divided into two settings, according to whether the adversary has access to the private key of one of the parties (the sender or the recipient) or not. 

First, \emph{outsider security} ensures that, if $S$ and $R$ are honest parties, then their channel is private and authenticated against adversaries from whom $\sdk_S$ and $\sdk_R$ are kept secret. Second, \emph{insider security} ensures that, if one party's secret key is compromised, the security guarantees ensured by the other party's secret key still hold. If the sender key is revealed, the receiver should still enjoy full privacy. The other case is quantumly unachievable: impossibility for signatures implies that releasing the receiver's secret key results in a forgeable (hence useless) scheme. In particular, quantum signcryption cannot provide \emph{non-repudiation.}

Formally defining quantum signcryption security is a challenge, as it is analogous to authenticated encryption, a notion which is troublesome in the quantum setting. The usual definitions require the ability to copy and compare previous queries with outputs of the adversary (e.g., to see if the adversary is attempting to decrypt the challenge.) Fortunately, a recent approach~\cite{AGM18} shows the way forward on such definitions. The idea is to ``split'' the standard security game into two experiments: an ``unrestricted'' experiment and a ``cheat-detecting'' experiment. In the unrestricted experiment, the adversary freely interacts with $\SigEnc$ and $\VerDec$ oracles. In the cheat-detecting experiment, the modified signcryption oracle signcrypts half of a maximally entangled state and stores the other half together with the adversary's plaintext. The modified verified-decryption oracle decrypts the input ciphertext, and checks if it is maximally entangled with one of the stored registers; if it is, it returns the corresponding stored plaintext.

Following this approach, we define outsider signcryption security 
through a pair of experiments, or games, 
as follows.

\begin{experiment}\label{exp:outsider-real-summary}
The \textbf{real outsider experiment} $\OutReal(\Pi, \algo A, n, S, R)$:
	\begin{algorithmic}[1]
		\State \textbf{output} $\algo A^{\SigEnc_{S, R}, \VerDec_{S, R}}(1^n)$.
	\end{algorithmic}
\end{experiment}

\vspace{-5pt}
\begin{experiment}\label{exp:outsider-ideal-summary}
The \textbf{ideal outsider experiment} $\OutIdeal(\Pi, \algo A, n, S, R)$:
	\begin{algorithmic}[1]
		\State define channel $E_{M \rightarrow C}$: prepare maximally entangled state $\ket{\phi^+}_{M'M''}$, store $(M'', M)$ in a set $\mathcal M$;
			return $\SigEnc_{S, R}$ applied to $M'$.
		\State define channel $D_{C \rightarrow M}$: \begin{itemize}
			\item[-] apply $\VerDec_{S, R}$ to $C$, place results in $M'$
			\item[-]  $\forall (M'', M) \in \mathcal M$: measure if $M'M''$ are max. entangled; if so return $M$
			\item[-] return $\bot$.
		\end{itemize}
		\State \textbf{output} $\adver^{E,D}(1^n)$.
	\end{algorithmic}
\end{experiment}

\begin{defn}
A \QSC $\Pi$ is \textbf{outsider secure} if for all QPTs $\adver$,
 	\begin{equation}
 	\left|\Pr\left[\OutReal(\Pi, \algo A, n) \to \real \right] -
 	 	\Pr\left[\OutIdeal(\Pi, \algo A, n) \to \real \right]\right| \leq \negl(n).\nonumber
 	\end{equation}
\end{defn}

For (sender-compromised) insider security, we demand \QINDCCAA security of the induced public-key encryption scheme, adapting the approach of~\cite{AGM18}.

\subsubsection{Constructions, security proofs}

Finally, we give secure constructions for quantum signcryption. Along the way, we also give several new results of independent interest. These results ``bootstrap'' classical security properties to the quantum setting via a generic classical-quantum hybrid construction $\CQHybGA[\Pi, \Sigma]$ combining a classical scheme $\Pi$ with a quantum scheme $\Sigma$. The key (or keys) for $\CQHybGA[\Pi, \Sigma]$ are the same as for $\Pi$. To encrypt $\rho$, we generate a fresh key $k$ for $\Sigma$ and encrypt $\rho$ with $k$. Furthermore, we ``encapsulate'' $k$ by encrypting it with $\Pi$ and attaching the resulting classical encryption to the encryption of $\rho$. Decryption first decapsulates the $\Sigma$-key $k$, then uses it to decrypt the rest. 

\begin{theorem}[informal]
	Let $\Sigma$ be one-time quantum authenticating (\cQCA~\cite{AGM18}.)
	\begin{enumerate}
		\item If $\Pi$ is a post-quantum \INDCCAA SKE/PKE scheme, then $\CQHybGA[\Pi, \Sigma]$ is a \QINDCCAA quantum SKE/PKE scheme.
		\item If $\Pi$ is a post-quantum AE scheme, then $\CQHybGA[\Pi, \Sigma]$ is a quantum AE scheme.
		\item If $\Pi$ is a post-quantum classical outsider-secure and insider-secure signcryption scheme, then $\CQHybGA[\Pi, \Sigma]$ is an outsider- and insider-secure quantum signcryption scheme.
	\end{enumerate}
\end{theorem}

\subsection{Related work}

As discussed above, the only previous work which directly addressed our topics of interest\footnote{There has been work on signing \emph{classical messages} by quantum means, starting from a work of Gottesman and Chuang~\cite{Gottesman2001}.} is Barnum et al.'s work on authentication of quantum states~\cite{BCG+02}. In the realm of signature impossibility, we improve the state of affairs from~\cite{BCG+02} significantly: we give formal definitions for correctness and security and a rigorous proof of a much stronger impossibility result. Indeed, \emph{only classical signatures are possible} is a much stronger statement (with far greater implications) than \emph{fully-secure signing for all quantum states is impossible.} In the realm of possibility, we develop a thorough foundation for the theory of quantum signcryption, which was briefly suggested (and called ``public key quantum authentication'') in~\cite{BCG+02}. In particular, we give formal definitions for correctness and security, examine insider and outsider attack scenarios, and give constructions and security proofs in the many-time setting (including CPA and CCA security.) 

Our approach to defining security for quantum signcryption expands the ideas of~\cite{AGM18} to the public-key setting. We give a new characterization lemma for encryption schemes which enables this expansion. We also give a generalization of the security proofs of~\cite{AGM18, ABF+16,BJ15}, establishing that classical-quantum ``hybrid schemes'' inherit the security properties of their component schemes in a very broad setting, including symmetric-key authenticated encryption, public-key encryption, and signcryption.

The security notions from~\cite{AGM18} were recently examined from a composable security point of view~\cite{BMPZ19}. The conclusion of that work was that the ``real/ideal'' definition of Quantum Authenticated Encryption from~\cite{AGM18} does guarantee composable security. This implies that the strongest notions of quantum singcryption security in this paper are composable as well.

\section{Preliminaries}

\subsection{Notation and Conventions}\label{sec:prelnot}
We will largely use the conventions regarding quantum information from~\cite{NC11}. We use $\hi_M$ to denote a complex Hilbert space with label $M$ and finite dimension $\dim M$. 
A quantum register is a physical system whose set of valid states is $\states(\hi_M)$. In our setting, a ``quantum register $M$'' is in fact an infinite family of registers $\{M_n\}_{n \in \NN}$ consisting of $p(n)$ qubits, where $p$ is some fixed polynomial. The notation $\tau_M$ will mean the maximally mixed state (i.e., uniform classical distribution) $\one / \dim M$ on $M$. We denote the euclidean norm of a vector, the trace norm of a matrix,  and the diamond norm of an operator on a matrix space, by  $\|\cdot\|_2$, $\|\cdot\|_1$ and $\|\cdot\|_\diamond$, respectively.  We set down notation for other norms as we use them.
By ``QPT'' we mean a polynomial-time uniform family of quantum circuits, i.e., an efficient quantum algorithm. Quantum algorithms implement completely positive (CP), trace-preserving (TP) maps, i.e., quantum channels. To indicate that $\Phi$ is a channel from register $A$ to $B$, we will write $\Phi_{A \to B}$. We often drop tensor products with the identity, e.g., given a map $\Psi_{BC \to D}$, we will write $\Psi \circ \Phi$ to denote the map $\Psi \circ (\Phi \otimes \id_C)$ from $AC$ to $D$.
For us, all classical schemes will be {\em post-quantum} secure, i.e., secure against QPT adversaries.

\subsection{Quantum and Classical Cryptographic Primitives}\label{sec:prelcrypto}
We use \secpar to denote the security parameter, 
and we write it as \secparam in unary form. 
A function is {\em negligible} if it grows asymptotically slower than $1/p(\secpar)$ for every polynomial $p$. The notation $x \rand X$ means that $x$ is sampled uniformly at random from the set $X$.
 
We assume basic familiarity with classical 
symmetric-key encryption (\SKE), 
public-key encryption (\PKE), digital signatures (\DS), and signcryption (\SC.) Recall that one can construct an outsider- and insider-secure signcryption scheme by combining a \DS with a \CCAAA-secure \PKE via ``encrypt then sign''~\cite{ADR02}. We adopt the following notation for keys: 
(i.) $k$ = symmetric key, 
(ii.) $\dk$ = secret decryption key, 
(iii.) $\ek$ = public encryption key, 
(iv.) $\sk$ = secret signing key, 
(v.) $\vk$ = public verification key, 
(vi.) $\sdk$ = secret signing and decryption key, and 
(vii.) $\vek$ = public verification and encryption key. 
For quantum encryption, we follow~\cite{ABF+16,AGM18}.

\begin{defn}\label{def:SKQE}
A symmetric-key quantum encryption scheme (or $\SKQE$) is defined as a triple of QPT algorithms:
\begin{enumerate}
\item (key generation) $\KeyGen:$ on input $\secparam$, outputs $k \rand \K$
\item (encryption) $\Enc: \K \times \states(\H_M) \rightarrow \states(\H_C)$
\item (decryption) $\Dec: \K \times \states(\H_C) \rightarrow \states(\H_M \oplus\ket{\bot})$
\end{enumerate}
such that $\| \Dec_k \circ \Enc_k - \id_M \oplus 0_\bot \|_\diamond \leq \negl(\secpar)$
for all $k \from \KeyGen(\secparam)$. 
\end{defn}

\begin{defn}\label{def:PKQE}
A public-key quantum encryption scheme (or $\PKQE$) is defined as a triple of QPT algorithms:
\begin{enumerate}
\item (key generation) $\KeyGen:$ on input $\secparam$, outputs $(\ek, \dk) \rand \EK \times \DK$
\item (encryption) $\Enc: \EK \times \states(\H_M) \rightarrow \states(\H_C)$
\item (decryption) $\Dec: \DK \times \states(\H_C) \rightarrow \states(\H_M \oplus\ket{\bot})$
\end{enumerate}
such that $\| \Dec_\dk \circ \Enc_\ek - \id_M \oplus 0_\bot \|_\diamond \leq \negl(\secpar)$
for all $(\ek, \dk) \from \KeyGen(\secparam)$. 
\end{defn}

All the spaces have $\poly(\secpar)$ bits or qubits; $\K$, $\EK$, $\DK$ are classical, and $C$, $M$ are quantum. 
We assume w.l.o.g. that \dk also includes \ek, and $k,\ek$ also include \secpar. 


\subsection{Quantum Authentication}\label{sec:DNSandQCA}

Given an attack map $\Lambda_{CB \to C\tilde B}$ on a scheme $\Pi = (\KeyGen, \Enc, \Dec)$ (where the adversary holds $B$ and $\tilde B$), we define the ``effective attack map'' by
$$
\Lambda^\Pi_{MB\to M\tilde B} := \mathbb{E} \left[ \Dec_k \circ \Lambda \circ \Enc_k\right],
$$
where the expectation is over the keys and classical randomness used by \Enc. Intuitively, a notion of authentication requires that, conditioned on acceptance, this map is the identity on $M$.

The standard one-time quantum authentication notion is ``\DNS''~\cite{DNS12}. It states that, for any attack map $\Lambda$, the effective attack map $\mathbb{E}_k [\Dec_k \circ \Lambda \circ \Enc_k]$ is equivalent to a combination of the identity and a ``discard'' map, along with some map on the adversary's private space. This provides for authentication of the plaintext only~\cite{AGM18}. A stronger notion, called quantum ciphertext authentication (\QCA,) prevents any adversarial modification of the ciphertext~\cite{AGM18} by placing further constraints on the simulator. Here, we recall both notions.

\begin{defn}[\cite{DNS12}]\label{def:DNS}
	A $\SKQE$ $\Pi = (\KeyGen, \Enc, \Dec)$ is \DNS-authenticating if, for any CP-map $\Lambda_{CB\to C\tilde B}$, there exist CP-maps $\Lambda^{\acc}_{B\to \tilde B}$ and $\Lambda^{\crej}_{B\to \tilde B}$ that sum to a TP map, such that:
	\begin{equation*}
		\left\|\Lambda^\Pi_{MB\to M\tilde B} - 
		\left(\id_M\otimes\Lambda^{\acc}_{B\to \tilde B} + \proj\bot_M\otimes \Lambda^{\crej}_{B\to \tilde B}\right)\right\|_\diamond\le\negl(n)\,.
	\end{equation*}
\end{defn}


\begin{defn}\label{def:qCauth}
	A \SKQE $\Pi=(\KeyGen, \Enc, \Dec)$ is \emph{ciphertext authenticating}, or \QCA, if for all CP-maps $\Lambda_{CB\to C\tilde B}$, there exists a CP-map $\Lambda^{\crej}_{B\to \tilde B}$ such that:
	\begin{equation*}\label{eq:normal-constraint}
		\left\|\Lambda^\Pi_{MB\to M\tilde B} - 
		\left(\id_M\otimes\Lambda^{\acc}_{B\to \tilde B} + \proj\bot_M\otimes \Lambda^{\crej}_{B\to \tilde B}\right)\right\|_\diamond\le\negl(n),
	\end{equation*}
	and $\Lambda^{\acc}_{B\to \tilde B}+\Lambda^{\crej}_{B\to \tilde B}$ is TP. Here
	$\Lambda^{\acc}_{B\to \tilde B}$ is given by:
	\begin{equation*}\label{eq:extra-constraint}
		\Lambda^{\acc}_{B\to \tilde B} (Z_B) =\mathbb{E}_{k,r}\left[\bra{\Phi_{k, r}} V_k^\dagger\Lambda \left(\Enc_{k;r}\left(\phi^+_{MM'} \otimes Z_B \right)\right)V_k\ket{\Phi_{k, r}}\right]
	\end{equation*}
	where 
	$\ket{\Phi_{k, r}} = \ket{\phi^+}_{MM'}\otimes\ket{\psi^{(k,r)}}_T.$
\end{defn}

While \DNS and \QCA are information-theoretic notions, it is straightforward to define variants which only require security against QPTs; we denote these by \cDNS and \cQCA, respectively. In these variants, as well as for security notions that we define in this paper, the condition on the diamond norm is relaxed by only requiring that no QPT adversary can distinguish the real and simulated channels with more than negligible advantage over guessing, i.e., the two channels are {\em computationally indistinguishable}.

All of these notions can be achieved with a simple scheme, as follows~\cite{ABE10,DNS12,AGM18}. The key selects a random element $C_k$ of the Clifford group. Encryption maps $\rho$ to $C_k \left( \rho \otimes \egoketbra{0^\secpar} \right) C_k^\dagger$. Decryption undoes $C_k$, and rejects if any of the attached qubits yield a non-zero measurement.



\subsection{Quantum Secrecy}\label{sec:prelqsecr}

For encryption schemes,
basic secrecy is straightforward to generalize to the quantum case, and yields the \QIND, \QINDCPA, and \QINDCCA notions~\cite{ABF+16,BJ15} 
which are quantum analogues of the corresponding classical security notions. 
The quantum counterparts of stronger 
notions like \INDCCAA or authenticated encryption (\CAE) require much more care; fortunately, a recent work showed that a solution is possible~\cite{AGM18}. 

For \SKQE schemes, we follow~\cite{AGM18} for the notion of \emph{quantum authenticated encryption (\QAE)}, 
while for \PKQE schemes, we formally define the 
novel 
notion of \emph{quantum indistinguishability under adaptive chosen-ciphertext attack (\QINDCCAA)} in~\expref{Section}{suppl:fullCCA2}. 
However, 
to simplify exposition, we will focus on the (weaker) \emph{plaintext security} 
(\cDNS-based) variants of 
these two notions. 
This means that we only check whether the adversary is attempting to `cheat' on the level of plaintexts, as in \expref{Experiment}{exp:outsider-ideal-summary}. We attach a ``\textsf{w}'' to the acronyms to indicate this distinction, and define these notions formally in~\expref{Section}{sec:weakqauth}. 
We use this approach for ease of exposition only: 
In \expref{Section}{sec:ctxt}, we describe 
how to extend all our results on signcryption to the full, ciphertext-secure 
setting.

\section{Weak (\cDNS-Based) Quantum Authentication for Encryption}\label{sec:weakqauth}


\subsection{Weak \CCAA Security for Public-Key Quantum Encryption}\label{sec:PKEwCCA2}

In~\expref{Section}{suppl:fullCCA2}
we 
first formally 
define CCA2 security for \PKQE schemes. Therein, we follow the route of symmetric-key quantum encryption schemes by adopting techniques from \cite{AGM18}. 

However, 
in the present work, for ease of exposition, we 
start by using 
integrity check techniques based on the simpler and more familiar \DNS notion of plaintext authentication instead of the more refined notion of \QCA from \cite{AGM18}. Therefore, 
in this section 
we define a \cDNS-style version of \QINDCCAA, which we call {\em weak \CCAA (\QINDwCCAA)}. Here, intuitively, the decryption oracle refuses to decrypt any ciphertext that decrypts to the challenge plaintext. This, in particular, has the advantage that the analysis in this section applies to all quantum encryption schemes (see \cite{AGM18} for details). On the other hand, this approach requires the repeated application of the Gentle Measurement Lemma \cite{Winter99}, and can therefore only be fulfilled by schemes with plaintext spaces of large dimension, with a loss in the tightness of the resulting reductions. While not providing the best possible quantum counterpart to \INDCCAA, this notion fits the simplified framework of this work, and the upgrade to \QINDCCAA via the \QCA-based setting is straightforward as discussed in~\expref{Section}{sec:ctxt}. See~\expref{Section}{suppl:fullCCA2} for a formal definition of the upgraded notion.

We define \QINDwCCAA in terms of adversaries playing two games. The first game is the same as \QINDCCAAtest from~\cite{AGM18} for the symmetric-key case, except that the \Enc oracle is replaced by the public key; in this game there are no restrictions on the use of $\Dec_\dk$ by $\algo A_2$. In the cheat-detection game \QINDwCCAAfake, instead, the adversary is declared to \cheat whenever he replays the challenge state (and declared to \cheat anyway with probability 1/2 otherwise). The only difference in respect to \QINDCCAAfake from~\cite{AGM18} (beyond the replacement of the \Enc oracle with the public key), is that the challenge-replay test is done \cDNS-style (detecting plaintext replay) instead of \cQCA-style (detecting ciphertext replay), as discussed above.
 
\begin{experiment}\label{exp:INDwCCAA-test}
The $\QINDCCAAtest(\Pi, \adver, \secpar)$ experiment:
	\begin{algorithmic}[1]
		\State $\chall$ runs $(\dk,\ek) \from \KeyGen(\secparam)$ and flips a coin $b \inrand \bit$;
		\State $\algo A_1$ receives $\ek$ and oracle access to $\Dec_\dk$;
		\State $\algo A_1$ prepares a side register $S$, and sends to $\chall$ a challenge register $M$;
		\State $\chall$ puts into $C$ either $\Enc_\ek(M)$ (if $b=0$) or $\Enc_\ek(\tau_M)$ (if $b=1$);
		\State $\algo A_2$ receives registers $C$ and $S$ and oracle access to $\Dec_\dk$;
		\State $\algo A_2$ outputs a bit $b'$.  \textbf{If} $b'=b$, \textbf{output} \win; otherwise \textbf{output} \rej.
	\end{algorithmic}
\end{experiment}

\begin{experiment}
The $\QINDwCCAAfake(\Pi, \adver,\secpar)$ experiment:
	\begin{algorithmic}[1]
		\State $\chall$ runs $(\dk,\ek) \from \KeyGen(\secparam)$;
		\State $\algo A_1$ receives $\ek$ and oracle access to $\Dec_\dk$;
		\State $\algo A_1$ prepares a side register $S$, and sends to $\chall$ a challenge register $M$;
		\State $\chall$ discards $M$, prepares $\ket{\phi^+}_{M'M''}$, and stores $M''$; then $\chall$ encrypts $M'$  (using $\ek$) and puts the resulting ciphertext into $C'$;
		\State $\algo A_2$ receives registers $C'$ and $S$ and oracle access to $D_\dk$, where $D_\dk$ is defined as follows. On input a register $C$:
		\begin{enumerate}[(1)]
			\State $\chall$ applies $\Dec_\dk$ to $C$, places result in $M$; 
			\State $\chall$ applies $\{\Pi^+, \one - \Pi^+\}$ to $M M''$;
			\State \algorithmicif\  {the outcome is $1$} \algorithmicthen: abort and \textbf{output} \textsf{cheat}; \algorithmicelse \ \textbf{return} $M$;
		\end{enumerate}
		\State $\chall$ draws a bit $b$ at random. \textbf{If} $b=1$, \textbf{output} \cheat; if $b=0$ \textbf{output} \rej.
	\end{algorithmic}
\end{experiment}

We define \QINDwCCAA in terms of the advantage gap of adversaries between the two games.

\begin{defn}\label{def:qINDwCCAA}
	A $\PKQE$ $\Pi$ is $\QINDwCCAA$ if, for all QPT adversaries $\adver$,
	$$
	\Pr[\QINDCCAAtest(\Pi, \adver, \secpar) \to \win] - \Pr[\QINDwCCAAfake(\Pi, \adver, \secpar) \to \cheat] \leq \negl(n)\,.
	$$
\end{defn}

Clearly, because an adversary for \QINDwCCAA is also an adversary for \QINDCCA with additional power and additional constraints on the scheme itself, we see that \QINDwCCAA implies \QINDCCA and hence also \QINDCPA.

\begin{prop}\label{prop:qINDwCCAAimplications}
	Let $\Pi$ be a $\QINDwCCAA$ \PKQE. Then $\Pi$ is also \QINDCCA and \QINDCPA.
\end{prop}

\subsection{Weak \QAE for Symmetric-Key Quantum Encryption}\label{app:wQAE}

Here 
we define a ``weak'' (\cDNS-style) version of the notion of quantum authenticated encryption \QAE for \SKQE schemes, that we denote as \wQAE. This will be used in~\expref{Section}{sec:many-time}. 
We start from the \QAEreal and \QAEideal games defined in \cite{AGM18} for \QAE, but 
in the latter game we 
replace the \cQCA-style replay check by a \cDNS-style replay check. We call the new game 
\wQAEideal.

\begin{experiment}\label{exp:wQAEreal}
	The experiment $\QAEreal(\Pi, \adver, \secpar)$:
	\begin{algorithmic}[1]
		\State $k \from \KeyGen(\secparam)$ ;
		\State \textbf{output} $\adver^{\Enc_k, \Dec_k}(1^\secpar)$.
	\end{algorithmic}
\end{experiment}

\begin{experiment}\label{exp:wQAEideal}
	The experiment $\wQAEideal(\Pi, \adver, \secpar)$:
	\begin{algorithmic}[1]
		\State $k \from \KeyGen(\secparam)$ ;
		\State define channel $E_{M \rightarrow C}$ as follows:
		\begin{enumerate}[(1)]
			\item prepare $\ket{\phi^+}_{M'M''}$, store $(M'', M)$ in a set $\mathcal M$;
			\item apply $\Enc_k$ to $M'$; \textbf{return} result.
		\end{enumerate}
		\State define channel $D_{C \rightarrow M}$ as follows:
		\begin{enumerate}[(1)]
			\item apply $\Dec_k$ to $C$, place results in $M'$; 
			\For {\textbf{each} $(M'', M) \in \mathcal M$}: 
			\State apply $\{\Pi^+, \one - \Pi^+\}$ to $M'M''$; \textbf{if} {outcome is $0$}: \textbf{return} $M$;
			\EndFor
			\State	\textbf{return} $\egoketbra{\bot}$;
		\end{enumerate}
		\State \textbf{output} $\adver^{E,D}(\secparam)$.
	\end{algorithmic}
\end{experiment}

A \SKQE is now defined to be {\em weak quantum authenticated encryption} (\wQAE) 
if no QPT adversary can reliably distinguish the two.

\begin{defn}
	A \SKQE $\Pi$ is \wQAE if for all QPT adversaries $\adver$ it holds:
	\begin{equation*}
		\left|\Pr\left[\QAEreal(\Pi, \adver, \secpar) \to \real \right]-\Pr\left[\wQAEideal(\Pi, \adver,\secpar) \to \real \right]\right| \leq \negl(\secpar).
	\end{equation*}
\end{defn}

Analogously we say that a classical \SKE scheme is \wCAE secure if it fulfills the classical restriction of the above (see \cite{AGM18}), defined by games \CAEreal and \wCAEideal, following the real vs ideal characterization of \CAE by Shrimpton \cite{Shrimpton04}. Clearly, because an adversary for \wQAE is also an adversary for \QINDCCA with additional power and additional constraints on the scheme itself, we see that \wQAE implies \QINDCCA and hence also \QINDCPA.

\begin{prop}\label{prop:wQAEimplications}
	Let $\Pi$ be a $\wQAE$ \SKQE. Then $\Pi$ is also \QINDCCA and \QINDCPA.
\end{prop}

\section{Security of Hybrid Quantum Encryption}\label{sec:hybqenc}

\subsection{Basic Definitions}\label{sec:hybqencdef}

In this section we study the security of 
canonical \SKQE and \PKQE schemes constructed in a hybrid classical-quantum way. 
Namely, 
we will use a classical \SKE (or \PKE) scheme to encrypt a random one-time key, and then use that key to encrypt a quantum state with a (usually one-time secure) 
\SKQE.
We denote the resulting construction of symmetric- (resp., public-) key quantum encryption by $\PiKEMS$ (resp., \PiKEMP).

\begin{construction}\label{cons:SKQE-generic}
Let $\Pi^\SKE$ be a symmetric-key classical encryption scheme, and $\Pi^\SKQE$ a symmetric-key quantum encryption scheme. We define 
the canonical hybrid 
symmetric-key quantum encryption scheme $\PiKEMS[\Pi^\SKE, \Pi^\SKQE] = (\KeyGen, \Enc, \Dec)$ as follows:
\begin{enumerate}
	\item $\KeyGen(\secparam):$ output $k \from \KeyGen^\SKE(\secparam)$;
	\item $\Enc:$ on input $\rho_M$, generate $k' \from \KeyGen^\SKQE(\secparam)$ and output\\ $(\Enc_k^\SKE(k'), \Enc^\SKQE_{k'}(\rho))$;
	\item $\Dec:$ on input $(c, \sigma_C)$, set $k' = \Dec_k^\SKE(c)$; if $k' = \bot$, output $\bot$ and otherwise output $\Dec^\SKQE_{k'}(\sigma)$.
\end{enumerate}
\end{construction}

\begin{construction}\label{cons:PKQE-generic}
Let $\Pi^\PKE$ be a public-key classical encryption scheme, and $\Pi^\SKQE$ a symmetric-key quantum encryption scheme. We define 
the canonical hybrid 
public-key quantum encryption scheme $\PiKEMP[\Pi^\PKE, \Pi^\SKQE] = (\KeyGen, \Enc, \Dec)$ as follows:
\begin{enumerate}
	\item $\KeyGen(\secparam):$ output $(\ek,\dk) \from \KeyGen^\PKE(\secparam)$;
	\item $\Enc:$ on input $\rho_M$, generate $k' \from \KeyGen^\SKQE(\secparam)$ and output\\ $(\Enc_\ek^\PKE(k'), \Enc^\SKQE_{k'}(\rho))$;
	\item $\Dec:$ on input $(c, \sigma_C)$, set $k' = \Dec_\dk^\PKE(c)$; if $k' = \bot$, output $\bot$ and otherwise output $\Dec^\SKQE_{k'}(\sigma)$.
\end{enumerate}
\end{construction}

\subsection{Security of Hybrid Symmetric-Key Encryption}\label{sec:sechybSKQE}


We now 
show that the hybrid construction for symmetric-key schemes, \PiKEMS, can be used to lift the security of a \SKQE scheme from \cQCA to \QAE, by combining it with a classical \CAE scheme. 

\begin{theorem}\label{thm:QAEhybrid}

Let $\Pi$ be a classical \CAE scheme, and let $\Sigma$ be a \cDNS- (resp., \cQCA-) secure symmetric-key quantum encryption scheme. Then $\PiKEMS[\Pi, \Sigma]$ is a \wQAE- (resp., \QAE-) secure symmetric-key quantum encryption scheme.
\end{theorem} 

\begin{proof}[Proof sketch] (For details, see \expref{Section}{sec:QAE-secu}). 
We only prove this theorem for the weak, plaintext-based version \wQAE starting from a \cDNS \SKQE scheme and a \CAE \SKE scheme. The general result follows from the discussion in \expref{Section}{sec:ctxt}.  IN this proof, indistinguishable means computationally indistinguishable.

We begin by defining a hybrid game \textsf{Hybrid $0$}, modifying the encryption and decryption routine. For encryption, we replace the encrypted one-time key $k'$ used for $\Enc^\Qu$ by an encryption of a freshly sampled key $k''\from\KeyGen^\Qu$ and store $(k',k'')$ in a database $\mathcal S$. For decryption, we first decrypt the classical part of the ciphertext and check whether the result is equal to one of the $k''$ in $\mathcal S$. If not, return $\bot$, if so, use the corresponding $k'$ to decrypt the quantum ciphertext and return the result. 
Suppose now that there exists an adversary \adver that can distinguish \QAEreal and \textsf{Hybrid $0$}. Then we can build an adversary \adver' that distinguishes the real and 
ideal worlds, \CAEreal and \CAEideal, in the characterization of \CAE. This is done by simulating the \QAEreal game played by \adver, assuming we are in the \CAEreal world. It turns out that this results in \adver' simulating \textsf{Hybrid $0$} when in the \CAEideal world. Hence \QAEreal and \textsf{Hybrid 0} are indistinguishable by the \CAE security of $\Pi^\Cl$.

We continue to show that the experiments \textsf{Hybrid $0$} and \wQAEideal are indistinguishable due to the \cDNS security of $\Pi^\Qu$.

We proceed using a standard hybrid method over the encryption queries of \adver. We define \textsf{Hybrid $i$} to replace the quantum plaintext by half of a maximally entangled state for the first $i$ queries, storing the other half together with the pair $(k', k'')$ and the plaintext. For the remaining queries it behaves like \textsf{Hybrid $0$}. The decryption routine behaves like in \textsf{Hybrid $0$}, unless the presented one-time key matches one from the first $i$ queries, then it performs the entanglement check measurement as in \wQAE.

We will now use the \cDNS security of $\Pi^\Qu$. Applying the definition of \cDNS to the outputs of an attack map $\Lambda_{C\to C^\ell}$ sequentially shows that \cDNS security also implies the existence of a simulator in the case of one encryption, but many decryptions. 
Here, the simulator has the form
\begin{equation}
	\Lambda^{\mathsf{Sim}}_{M\to M^\ell}=\sum_{i=1}^\ell \id_{M\to M_i}\otimes \left(\bot^{\otimes (\ell-1)}\right)_{M_{i^c}}\otimes \Lambda^{(i)}_{B\to B}+\left(\bot^{\otimes \ell}\right)_{M^\ell}\otimes \Lambda^\bot_{B\to B},
\end{equation}
where $B$ is the adversary's side information register. This is proven in detail in Lemma \ref{lem:multiDNS}. 
We can replace the combination "$i+1$-st encryption - attack - decryptions matching $k''_{i+1}$" by this simulator in both \textsf{Hybrids $i$} and $i+1$, yielding \textsf{Hybrids $i_a$} and $(i+1)_b$ which are indistinguishable from \textsf{Hybrids $i$} and $i+1$, respectively. It remains to observe that the entanglement tests in \textsf{Hybrid} $i+1$ exactly project onto the respective identity contribution of $\Lambda^{\mathsf{Sim}}$ (a similar idea to the Broadbent-Waynewright simulator \cite{BW16}), i.e.
\begin{equation}
	\bra{\phi^+}_{MM_i}\Lambda^{\mathsf{Sim}}_{M\to M^\ell}\left(\phi^+_{MM_i}\otimes (\cdot)_B\right)\ket{\phi^+}_{MM_i}=\left(\bot^{\otimes (\ell-1)}\right)_{M_{i^c}}\otimes \Lambda^{(i)}_{B\to B}.
\end{equation}
This immediately implies that \textsf{Hybrids $i_a$} and $(i+1)_b$ are identical.


We conclude that \textsf{Hybrid $0$} and \textsf{Hybrid $q$} are indistinguishable. It remains to show that \textsf{Hybrid $q$} and \wQAEideal are indistinguishable. This follows by Gentle Measurement \cite{Winter99}: every time the entanglement test is applied to the current plaintext and an unrelated purification, the state remains unperturbed. 
\end{proof}

\subsection{Security of Hybrid Public-Key Encryption}\label{sec:sechybPKQE}



In this section, we prove the security of the hybrid construction \PiKEMP under suitable building blocks. For \PKQE schemes, the notions of \QINDCPA and \QINDCCA are defined in \cite{ABF+16,BJ15}, where the adversary $\adver=(\algo A_1,\algo A_2)$ consists of two parts (pre-challenge and post-challenge), and is playing against a challenger $\chall$, which is a fixed algorithm determined only by the security game and the scheme. The security of the hybrid construction \PiKEMP for these two notions is trivial.

\begin{theorem}\label{thm:PKE-hybrid-full}
Let $\Pi^\PKE$ be an \INDCPA (resp., \INDCCA) \PKE. Let further ${\Pi^\SKQE}$ be a \QIND \SKQE. Then $\PiKEMP[\Pi^\PKE, {\Pi^\SKQE}]$ is a \QINDCPA (resp., \QINDCCA) \PKQE.
\end{theorem}

\begin{proof}
We will prove the \CPA case; the \CCA case is essentially identical. For a contradiction, let $\algo A$ be a QPT machine which wins the \QINDCPA game against $\PiKEMP[\Pi^\PKE, \Pi^\SKQE]$ with probability $1/2 + \delta$ where $\delta$ is 
non-negligible in $\secpar$. Let $\algo A_1$ and $\algo A_2$ be the pre-challenge and post-challenge algorithms of $\algo A$. Define an adversary $\algo A'$ against $\Pi^\PKE$ as follows. The pre-challenge algorithm $\algo A_1'$ accepts $\ek$, runs $\rho \from \algo A_1(\ek)$, and outputs $m \from \KeyGen^\SKQE(\secparam)$. The post-challenge algorithm $\algo A_2'$, on input $c$, proceeds as follows: (i.) flip a coin $b' \inrand \bit$, (ii.) if $b'=0$, set $t = \algo A_2(c, \Enc^{\Pi^\SKQE}_m(\rho))$; if $b'=1$ set $t = \algo A_2(c, \Enc^{\Pi^\SKQE}_m(\one/d_M))$; (iii.) if $t = b'$ set $b_\textsf{out} = 0$; otherwise set $b_\textsf{out} \inrand \bit$, and (iv.)output $b_\textsf{out}$.

Now consider $\algo A'$ in the \INDCPA game, and let $b$ be the challenge bit. If $b =0$, then the execution of $\algo A$ is exactly simulating the \QINDCPA game, and so $\Pr[b_\textsf{out} = b] \geq 1/2 + \delta$. If $b=1$, then the encryption of $\rho$ is done with an independent key; by \QIND, $\Pr[t = b'] = 1/2$ which implies that $\Pr[b_\textsf{out} = b] = 1/2$. We thus have that, overall, $\Pr[b_\textsf{out} = b] \geq 1/2 \cdot (1/2 + \delta) + 1/2 \cdot 1/2 \geq 1/2 + \delta/2$, meaning that $\algo A'$ would break the \INDCPA security of $\Pi^\PKE$.
\end{proof}

We now prove that the hybrid construction \PiKEMP also inherits (\cDNS-style) \CCAA security from the underlying \PKE, {\em provided} the underlying \SKQE is now not only \QIND, but also (\cDNS-style) authenticating.



\begin{theorem}\label{thm:KEMwCCA2}
Let $\Pi^\PKE$ be an \INDCCAA \PKE scheme, and let $\Pi^\SKQE$ be a \cDNS-secure \SKQE scheme. Then $\PiKEMP[\Pi^\PKE, \Pi^\SKQE]$ 
is a \QINDwCCAA-secure \PKQE scheme.
\end{theorem}
\begin{proof}[Proof sketch] (For details, see \expref{Section}{sec:proofwCCA2}). We follow the same strategy as the proof of \QINDCPA security, i.e., \expref{Theorem}{thm:PKE-hybrid-full}. Specifically, given an adversary $\adver$ which can distinguish the ``test'' and ``fake'' games, we will build an adversary $\adver'$ against the CCA2-secure classical public-key scheme. This adversary proceeds as follows:

\noindent (1.) pass the input ($1^n$ and the public key $\ek$) to start $\adver$;

\noindent (2.) simulate each decryption query in the obvious way: query the PKE $\Dec$ oracle on the classical part of the input, then decrypt the quantum part using the  resultant plaintext as the \SKQE key.

\noindent (3.) When $\adver$ outputs the challenge $\rho$, send a fresh one-time key $k$ as our \INDCCAA challenge, receive the challenge PKE ciphertext $c$ back. 

\noindent (4.) Now flip a coin $b$ to decide whether, for the remainder of the game, we will simulate the ``test'' game ($b = 0$) or the ``fake'' game ($b=1$) with $\adver$. In the case ``test'', we will use the challenge ciphertext $(c, \Enc_k^\QCA(\rho))$. 
Clearly, in this case, we can faithfully simulate the rest of the game.

\noindent (5.) if $\adver$ correctly guesses $b$, we output $0$. If not, we output a fair coin.

We argue that $\adver'$ wins the PKE \INDCCAA game with non-negligible advantage over random guessing. First, if the PKE challenge bit is $0$ (i.e., undisturbed challenge), then we are faithfully simulating $\adver$ in either the ``test'' or the ``fake'' games. We will thus gain precisely the advantage of $\adver$ over random guessing, in this case. It remains to check that, if the PKE challenge bit is $1$ (i.e., discarded challenge), then $\adver'$ does \emph{no worse} than random guessing. This is done by reduction to \cDNS security. Note that, in this case, the \SKQE encryption provided to $\adver$ is performed with a key $k$ which is independent of all other random variables in the game (since $c$ is an encryption of a random string, and not $k$.) If the adversary can nonetheless distinguish the ``test'' game from the ``fake'' game, then (by the definition of these games) this implies that he cannot be simulated by an ``ignore or discard'' channel demanded by the \cDNS security definition. 
\end{proof}

The above theorem also extends to the stronger ciphertext security (\cQCA-based) case, yielding \QINDCCAA hybrid \PKQE, as defined in~\expref{Section}{suppl:fullCCA2}. We can now summarize all the results discussed in this section.

\begin{cor}\label{cor:PKE-hybrid}
Let $\Pi$ be a \PKE, and $\Sigma$ a \SKQE. Then $\PiKEMP[\Pi, \Sigma]$ is a public-key quantum encryption scheme (\PKQE). Moreover, we have:
\begin{enumerate}
\item If $\Pi$ is \INDCPA and $\Sigma$ is \QIND, then $\PiKEMP[\Pi, \Sigma]$ is \QINDCPA.
\item If $\Pi$ is \INDCCA and $\Sigma$ is \QIND, then $\PiKEMP[\Pi, \Sigma]$ is \QINDCCA.
\item If $\Pi$ is \INDCCAA and $\Sigma$ is \cDNS, then $\PiKEMP[\Pi, \Sigma]$ is \QINDwCCAA.
\item If $\Pi$ is \INDCCAA and $\Sigma$ is \cQCA, then $\PiKEMP[\Pi, \Sigma]$ is \QINDCCAA.
\end{enumerate}
\end{cor}

\section{Impossibility of Signing Quantum States}\label{sec:impossible}


\subsection{Defining Quantum Signatures}\label{sec:defqs}\label{sec:defQSC}

In this section we discuss how to properly define in a formal way what a quantum signature scheme should be. A first attempt would be to translate ``quantumly'' in the most natural way the usual definition of classical schemes; that is, trying to define the following.

\begin{defn}[Quantum Signature, Wrong Definition \#1]\label{def:QSSwrong1}
	A {\em quantum signature scheme (or $\QS$)} with signing-verifying key space $\{\calS_\secpar\}_\secpar \times \{\V_\secpar\}_\secpar$ (where $\{\calS_\secpar\}_\secpar$ and $\{\V_\secpar\}_\secpar$ are families of spaces of bitstrings of size polynomial in \secpar), message space $\P$ and signature space $\C$ (both being complex Hilbert spaces of fixed finite dimension) is a triple of QPT algorithms:
	\begin{enumerate}
		\item (key generation) $\KeyGen:$ on input \secparam, outputs signing-verifying key pair $(\sk,\vk) \in \calS_\secpar \times \V_\secpar$, where 
		we assume w.l.o.g. that \sk includes \vk, and \vk includes \secpar.
		\item (Sign) $\Sign$: on input a signing key $\sk \in \calS_\secpar$ and a quantum state $\mu \in \states(\P)$, outputs a quantum state (``quantum signature'') $\sigma \in \states(\C)$; we write this as $\sigma \from \Sign_\sk(\mu)$.
		\item (Verify) $\Ver$: on input a verification key $\vk \in \V_\secpar$, a quantum state $\mu \in \states(\P)$, and a quantum signature $\sigma \in \states(\C)$, outputs a verification bit $b \in \{\acc,\rej\}$; we write this as $b \from \Ver_\vk(\mu,\sigma)$.
	\end{enumerate}
	Moreover, the following correctness property must hold:
	$$
	\forall (\sk,\vk) \from \KeyGen(\secparam), \forall \mu \in \states(\P) \implies \\
	\Pr[\Ver_\vk(\mu,\Sign_\sk(\mu)) \to \rej ] \leq \negl(\secpar).
	$$
\end{defn}

Looking at the above definition, it is clear that there is an issue: the \Sign and \Ver procedures consume the same message state. 
This is impossible in the quantum setting, given the no-cloning theorem\footnote{We remark that the no-cloning theorem rules out \emph{generic} cloning of \emph{arbitrary} states; indeed, our interest is precisely in generic signing of arbitrary states. Cloning particular collections of known states is possible, e.g., if a method for preparing the states is known. This is uninteresting, as one can then simply sign the classical description of the preparation method.}. A natural ``fix'' would be to make sure that both the \Sign and \Ver algorithms output an unmodified copy of the message state. This would lead to the following.

\begin{defn}[Quantum Signature, Wrong Definition \#2]\label{def:QSSwrong2}
	\begin{sloppypar}
		A \emph{quantum signature scheme (or $\QS$)} with signing-verifying key space $\{\calS_\secpar\}_\secpar \times \{\V_\secpar\}_\secpar$ (where $\{\calS_\secpar\}_\secpar$ and $\{\V_\secpar\}_\secpar$ are family of spaces of bitstrings of size polynomial in \secpar), message space $\P$ and signature space $\C$ (both being complex Hilbert spaces of fixed finite dimension) is a triple of QPT algorithms:
	\end{sloppypar}
	\begin{enumerate}
		\item (key generation) $\KeyGen:$ on input \secparam, outputs signing-verifying key pair $(\sk,\vk) \in \calS_\secpar \times \V_\secpar$, where we assume w.l.o.g. that \sk includes \vk, and \vk includes \secpar.
		\item (Sign) $\Sign$: on input a signing key $\sk \in \calS_\secpar$ and a quantum state $\mu \in \states(\P)$, outputs two quantum states: a quantum message $\mu' \in \states(\P)$  and a quantum signature $\sigma \in \states(\C)$; we write this as $(\sigma,\mu') \from \Sign_\sk(\mu)$.
		\item (Verify) $\Ver$: on input a verification key $\vk \in \V_\secpar$, a quantum state $\mu \in \states(\P)$, and a quantum signature $\sigma \in \states(\C)$, outputs a quantum message $\mu' \in \states(\P)$  and a verification bit $b \in \{\acc,\rej\}$; we write this as $(\mu',b) \from \Ver_\vk(\mu,\sigma)$.
	\end{enumerate}
	Moreover, the following correctness properties must hold:
	$$
	\forall (\sk,\vk) \from \KeyGen(\secparam), \forall \mu \in \states(\P) \implies 
	$$
	\begin{enumerate}
		\item $\Pr[\Ver_\vk(\Sign_\sk(\mu)) \to (., \rej) ] \leq \negl(\secpar)$
		\item $\Sign_\sk(\mu) \to (\mu', .) \implies \| \mu - \mu'\|_\diamond \leq \negl(\secpar)$
		\item $\Ver_\vk(\mu,\sigma) \to (\mu', .) \implies \| \mu - \mu'\|_\diamond \leq \negl(\secpar)$
	\end{enumerate}
\end{defn}

It is easy to notice that this definition also has problems. First of all, the conditions on the signing map are too strong, as they imply (again by no-cloning) that the signature cannot be noticeably correlated with the message state itself. This would make the goal of achieving any reasonable security notion hopeless. The third correctness condition is also unnecessarily strong, as it requires that the original state is recovered even when verification fails. 

This finally leads to the following definition.

\begin{defn}\label{def:QSSnatural} A {\em \textbf{quantum signature scheme}} (or $\QS$) is a triple of QPTs:

\begin{enumerate}
\item (key generation) $\KeyGen(1^n):$ output $(\sk,\vk) \in \SK \times \VK$.

\item (sign) $\Sign: \SK \times \states(\H_M) \rightarrow \states(\H_C)$

\item (verify) $\Ver: \VK \times \states(\H_C) \rightarrow \states(\H_M \oplus \ket{\bot})$ \,.

\end{enumerate}
\end{defn}

Due to no-cloning, we relaxed the classical requirement that a signature is an additional string accompanying the message. The key space $\SK \times \VK$ is a family of sets of bitstrings of size $\poly(\secpar)$, and 
$M$ and $C$ are quantum registers of $\poly(\secpar)$ qubits. In a typical setting, the verification key $\vk$ will be public, while the signing key $\sk$ stays private. Since our goal is to show impossibility, we define correctness separately 
from the basic definition 
and 
allow for schemes with very weak guarantees. 

Recall that a projector $P$ defines a two-outcome measurement channel $N$ via
$
N(X)=\tr((\one-P)X)\proj{0}+\tr(PX)\proj{1}.
$

\begin{defn}\label{def:QSS-correct}
A \QS $(\KeyGen, \Sign, \Ver)$ is \textbf{correct} if it satisfies
\begin{equation}\label{eq:qss-correct}
\| \Ver_\vk \circ \Sign_\sk - \id_M \oplus 0_\bot \|_\diamond \leq \negl(n)
\end{equation}
for all $(\sk, \vk) \from \KeyGen(1^n)$. More generally, given a finite set $\mathcal N$ of two-outcome measurements on $\H_M \oplus \ket{\bot}$, we say that a $\QS$ is \textbf{$\mathcal N$-correct} if it satisfies
	$
	\| N \circ\Ver_{\vk} \circ \Sign_{\sk} - N\oplus 0_\bot  \|_\diamond \leq \negl(\secpar)
	$
for all $(\sk, \vk) \from \KeyGen(\secparam)$ and all $N \in \mathcal N$. 
\end{defn}

\begin{defn}\label{def:QSS-secure}
Let $\mathcal L$ be a finite set of two-outcome measurements on 
the space 
$\H_M \oplus \ket{\bot}$. A \QS is $\varepsilon$-one-time \textbf{$\mathcal L$-secure} if, for any QPT adversary $\A$, for all $L \in \mathcal L$ there exists a probability $p\in[0,1]$ such that
\begin{equation}\label{eq:DNS-weak}
\left\|L \circ\mathbb E\left[\Ver_{\vk} \circ \mathcal A(\vk)\circ \Sign_{\sk}\right] - p L +(1-p)\bot \right\|_\diamond \leq \varepsilon\,.
\end{equation}
\end{defn}

Here, $\bot$ denotes the reject map $X \mapsto \proj\bot \otimes \tr(X)$. 
The constraint \eqref{eq:DNS-weak} is essentially a weakened version of the \DNS authentication security definition~\cite{DNS10}, adapted to the public-key case. It is weakened in three ways: (i.) we only ask for computational security, (ii.) the adversary does not hold any side information about the plaintext, and (iii.) security only holds for the selected set of measurements. Note that correctness and security are not required to hold for the same set of measurements.
In any case, as we will show next, 
even such 
security cannot be achieved except for trivial choices of $\mathcal N$ and $\mathcal L$.



\subsection{Impossibility of Quantum Signatures}\label{sec:imp}\label{sec:superlemma}


First of all we recall (a quantum channel version of) Uhlmann's theorem \cite{Kretschmann2008}. In the following, we will denote the operator norm by $\|\cdot\|_\infty$.

\begin{theorem}[Theorem I in \cite{Kretschmann2008}]\label{thm:channel-uhlmann}
	Let $T^{(i)}_{A\to A}$, $i=1,2$ be quantum channels with the same input and output systems, and let $U^{{(i)}}_{A\to AE}$ be Stinespring dilation isometries of $T^{(i)}$. Then there exists a unitary $V_{E\to E}$ such that
	\begin{equation*}
	\|VU^{(1)}-U^{(2)}\|_\infty^2\le \|T^{(1)}-T^{(2)}\|_\diamond\le 2\|VU^{(1)}-U^{(2)}\|_\infty.
	\end{equation*}
\end{theorem}

To prove our strong impossibility result, we will need the following technical lemma characterizing quantum encryption.

\begin{lem}[generalization of Lemma B.9 in~\cite{AM17}]\label{lem:PKE-char}
	Let $\Pi=(\KeyGen, \Enc,$\\$ \Dec)$ be a $\PKE$ with exact (approximate) correctness. Then $\Enc$ and $\Dec$ have the following form, for every key pair $k = (\dk, \ek)$:
	\begin{align*}
	&\left\|\Enc_\ek- V_k\left((\cdot)\otimes (\sigma_k)_T\right)V_k^\dagger\right\|_\diamond\le\varepsilon\nonumber\\
	&\left\|\Dec_\dk( V_kP^{\sigma_k}_T V_k^\dagger(\cdot)V_kP^{\sigma_k}_T V_k^\dagger) - \tr_{T}\left[P^{\sigma_k}_T \left(V_k^\dagger (\cdot)V_k\right) P^{\sigma_k}_T \right] \right\|_\diamond\le\varepsilon.
	\end{align*}
	Here, $\sigma_k$ is a state on register $T$, $P^{\sigma_k}_T$ and $\bar P^{\sigma_k}_T$  is an orthogonal projector such that $\|P^{\sigma_k}_T\sigma_kP^{\sigma_k}_T-\sigma_k\|$ is  negligible, $V_k$ is a unitary operator and $\varepsilon$ is negligible.
\end{lem}
\begin{proof} To simplify notation, we fix a key pair $k$ and omit all key subscripts. We indicate input and output systems for channels, and support systems for operators, using subscripts as necessary. 
Let $(U )_{M\to CE}$, and $(W )_{C\to MF}$, be Stinespring dilation isometries of $\Enc$, and $\Dec$, respectively, i.e.,
$
\Enc (X)=\tr_{E}U  XU ^\dagger
$
and
$
\Dec (Y)=\tr_FW  YW ^\dagger.
$
Now, $\varepsilon$-approximate correctness together with \expref{Theorem}{thm:channel-uhlmann} implies that there exists a pure state $\ket{\phi }_{ EF}$ such that
	\begin{equation*}
		\left\|\left(W \right)_{C\to MF} \left(U \right)_{M\to CE}-\one_M\otimes \ket{\phi }_{ EF}\right\|_\infty\le \sqrt\varepsilon.
	\end{equation*}
Let $\hat W_{CR\to MF}$ be a unitary such that $\hat W_{CR\to MF}\ket 0_R=W_{C\to MF}$. Such a unitary exists without loss of generality (If $C$ does not divide $MF$, we can just pick a bigger $F$). By unitary invariance of operator norm, 
$
		\bigl\| U _{M\to CE}\otimes \ket 0_R-\hat W^\dagger_{ MF\to CR} \ket{\phi }_{ EF}\bigr\|_\infty\le \sqrt\varepsilon
$. Therefore we get
	\begin{equation}\label{eq:closetononiso}
	\bigl\| U _{M\to CE}-\bra 0_R\hat W^\dagger_{ MF\to CR} \ket{\phi }_{ EF}\bigr\|_\infty\le \sqrt\varepsilon.
	\end{equation}
	by the sub-multiplicativity of the operator norm. Let 
	\begin{equation*}
	\bra 0_R\hat W^\dagger_{ MF\to CR} \ket{\phi }_{ EF}=U^{(1)}_{M\to CE}D_MU^{(2)}_M
	\end{equation*}
	 be the singular value decomposition of the second matrix. \expref{Equation}{eq:closetononiso} implies
$
	 \left\| \one_M-D_M\right\|_\infty\le \sqrt\varepsilon,
$
	 and hence
	 \begin{align}\label{eq:makeiso}
	 &\bigl\| \bra 0_R\hat W^\dagger_{ MF\to CR} \ket{\phi }_{ EF}-\bra 0_R\bigl(\hat W^\dagger \bigl(U^{(2)}\bigr)^\dagger D^{-1}U^{(2)}\bigr)_{MF\to CR}\ket{\phi }_{ EF}\bigr\|_\infty\nonumber\\
	 \le &\bigl\| U^{(1)}_{M\to CE}D_MU^{(2)}_M-U^{(1)}_{M\to CE}U^{(2)}_M\bigr\|_\infty\le\sqrt\varepsilon
	 \end{align}
	and 
$
		\tilde V_{M\to CE}=\bra 0_R\bigl(\hat W^\dagger \bigl(U^{(2)}\bigr)^\dagger D^{-1}U^{(2)}\bigr)_{MF\to CR}\ket{\phi }_{ EF}
$
	 is an isometry. Combining Equations \eqref{eq:closetononiso} and \eqref{eq:makeiso}, we arrive at
	\begin{equation*}
	\bigl\| U _{M\to CE}-\tilde V _{M\to CE}\bigr\|_\infty\le 2\sqrt\varepsilon.
	\end{equation*}
	Defining 
	$
	\tilde \Enc (X)=\tr_E\tilde V  X\tilde V ^\dagger,
	$
	we conclude that
	\begin{align*}
		\bigl\|\Enc -\tilde \Enc \bigr\|_\diamond
		\le& \max_{\rho_{MS}}\bigl\|U \rho U ^\dagger-\tilde V \rho \tilde V ^\dagger\bigr\|_\diamond\le 4\sqrt\varepsilon,
	\end{align*}
	where the first inequality is due to the definition of the diamond norm and the fact that the trace norm is non-increasing under partial trace, and the second inequality is a double application of H\"older's inequality. $\tilde \Enc $ has the form we  want, although this fact is still quite hidden. To show it, we define
		\begin{equation*}
	A_{MF\to C}=\bra 0_R\bigl(\hat W^\dagger \bigl(U^{(2)}\bigr)^\dagger D^{-1}U^{(2)}\bigr)_{MF\to CR},
	\end{equation*}
	so that we can write
$
		\tilde \Enc (X)=\left(A \right)_{MF\to C}\left(X_M\otimes \left(\phi \right)_{ F}\right)\left(A ^\dagger\right)_{ C\to MF}.
$
	$A$ is not, in general, an isometry. $\tilde\Enc $ is, however, trace preserving, implying
$
		A ^\dagger A =\one_M\otimes \kappa_F
$
	such that $\tr(\kappa_F\phi_{ F})=1$. Setting $T=F$, $\sigma_T=\sqrt\kappa_F \phi_F \sqrt \kappa_F$ and letting $V $ be a completion of $A \kappa^{-1/2}$ to a unitary shows now that $\tilde\Enc $ has the desired form. For notational convenience, define $\delta=\varepsilon+4\sqrt{\varepsilon}$. 
	
	To show the form of the decryption map, observe that, again by \expref{Theorem}{thm:channel-uhlmann}, there exists a possibly different Stinespring dilation isometry $W'_{C\to CT}$ for $\Dec$ and a 
state $\ket{\phi'_{k}}_{ET}$ such that 
	\begin{equation*}
		\left\|W'_{C\to CT}V_{MT\to C}\ket\phi_{TE}-\one_M\otimes\ket{\phi'}_{ET}\right\|_\infty\le\sqrt{\delta}
	\end{equation*}
	In particular, there is a unitary $U_T$ such that $\|U_T\ket\phi_{TE}-\ket{\phi'}_{TE}\|_2\le\sqrt{\delta}$, so
	\begin{equation}\label{eq:opnormineq}
	\left\|W'_{C\to CT}V'_{MT\to C}\ket{\phi'}_{TE}-\one_M\otimes\ket{\phi'}_{TE}\right\|_\infty\le 2\sqrt{\delta},
	\end{equation}
	with $V'=VU^\dagger$. Let 
$
		\ket{\phi'}_{TE}=\sum_i\sqrt q_i \ket{\gamma_i}_T\otimes \ket{\eta_i}_E
$
	be the Schmidt decomposition of $\ket{\phi'}$. \expref{Equation}{eq:opnormineq} implies that
	\begin{equation}\label{eq:opnormineq2}
		q_i\left\|W'_{C\to CT}V'_{MT\to C}\ket{\gamma_i}_{T}-\one_M\otimes\ket{\gamma_i}_{T}\right\|_\infty\le 2\sqrt{\delta}.
	\end{equation}
	Let 
$
		P_T=\sum_{i: q_i\ge(4\delta)^{1/6}}\proj{\gamma_i}.
$
	Then we get
	\begin{align*}
		&\left\|W'_{C\to CT}V'_{MT\to C}P_T(V'^\dagger)_{C\to MT}-P_T(V'^\dagger)_{C\to MT}\right\|_\infty\nonumber\\
		\le& \sum_{\substack{i: q_i\ge(4\delta)^{\frac 1 6}}}\left\|W'V'\proj{\gamma_i}_{T}V'^\dagger-\one_M\otimes\proj{\gamma_i}_{T}V'^\dagger\right\|_\infty
		\le(4\delta)^\frac 1 6, \nonumber
	\end{align*}
	where we have used \expref{Equation}{eq:opnormineq2} and the fact that $\mathrm{rk}P\le\frac{1}{(4\delta)^\frac 1 6}$. Using that $W'$ is a Stinespring dilation isometry for $\Dec$ shows the second claimed inequality.
\end{proof}

Note that the behavior of $\Dec_\dk$ can be arbitrary outside of the range of $V_kP^{\sigma_k}_T V_k^\dagger$. The standard behavior is to output the reject symbol $\bot$ when presented with an invalid ciphertext. This lemma gives an information-theoretic characterization of the channels $\Enc_\ek, \Dec_\dk$; this is why it's ok that $V$ and $\sigma$ are indexed by the keypair. This does not guarantee that $\Enc_\ek$ can be efficiently implemented in this form: knowledge of $\ek$ alone is not enough to efficiently implement $\sigma_k$ and $V_k$ (otherwise it would also be enough to decrypt!). In actual schemes, $\ek$ is enough to efficiently implement $\sigma_k$ and $V_k$ \emph{on relevant inputs for encryption}\footnote{e.g. because $\sigma_k$ is diagonal and represents a classical random variable, and $V_k$ maps computational basis states to computational basis states in a way that can be computed efficiently, but not in a coherent (or garbageless, reversible) way.} . By viewing $\sigma_\dk$ in its eigenbasis, we see that the characterization can be enhanced so it only uses classical randomness.

\begin{cor}\label{cor:randomized-isometry}
Let $\Pi=(\KeyGen, \Enc, \Dec)$ be a $\PKQE$. Then for every key pair $k := (\ek, \dk)$, there exists a probability distribution $p_k:\{0,1\}^t\to[0,1]$ and a family of quantum states $\ket{\psi^{(k,r)}}_T$ such that $\Enc_k$ is equivalent to the following algorithm:
(i.) sample $r \in \bit^t$ according to $p_k$; (ii.) apply the map $X_M \mapsto V_k \left(X_M \otimes \proj{\psi^{(k,r)}}_T \right) V_k^{\dagger}$.
Here $V_k$ and $T$ are as in \expref{Lemma}{lem:PKE-char}, and $t$ is the number of qubits in $T$.
\end{cor}

Importantly, even if $\Enc_k$ is a PPT algorithm, the functionally-equivalent algorithm provided by \expref{Corollary}{cor:randomized-isometry} may not be. We thus define the following.

\begin{condition}\label{con:efficient}
{\em Let $\Pi$ be a \PKQE, and let $p_k$, $\ket{\psi^{(k, r)}}$ and $V_k$ be as in \expref{Corollary}{cor:randomized-isometry}. We say {\em $\Pi$ satisfies \expref{Condition}{con:efficient}} if there exist QPTs for (i.) sampling from $p_k$, (ii.) preparing $\ket{\psi^{(k, r)}}$, and (iii.) implementing $V_k$ on inputs of the form $\rho \otimes \proj{\psi^{(k, r)}}$, and this holds for all but a negligible fraction of $k$ and $r$.} 
\end{condition}

We stress that we are not aware of any known or trivially defined construction of \PKQE where \expref{Condition}{con:efficient} does not hold.

We will now show that a \QS cannot be both correct and secure for a pair of measurements unless those measurements commute, in the sense that their sequential application yields the same outcome distribution no matter the order. Recall that a measurement $M$ also defines an instrument, i.e., a measurement map that keeps the post-measurement state, by
\begin{equation}
	\tilde M(X)=(\one-P)X(\one-P)\otimes \proj{0}+PXP\otimes \proj 1
\end{equation}In the following theorem, the composition of measurements is understood to mean that the instrument of the first measurement is applied and the the second measurement acts on the post-measurement state of the first one, i.e.
\begin{equation}
	( M_1\circ  M_0)_{A\to R_0 R_1}(X_A):=( M_1)_{A\to R_1}\left( (\tilde M_0)_{A\to AR_0}(X_A)\right).
\end{equation}
The index of the outcome register $R$ is tied to the index of the instrument, not to the order the measurements are performed, i.e.
\begin{equation}
( M_0\circ  M_1)_{A\to R_0R_1}(X_A)=(M_0)_{A\to R_0}\left( (\tilde M_1)_{A\to AR_1}(X_A)\right).
\end{equation}


\begin{theorem}\label{thm:imp-1}
	Let $\Pi = (\KeyGen, \Sign, \Ver)$ be a \QS and $\{(M_0)_{A\to R_0},(M_1)_{A\to R_1}\}$ be two-outcome projective measurements with efficiently implementable instruments $\tilde M_i$. If $\Pi$ is $\{ M_0, M_1 \}$-correct and $\varepsilon$-one-time $\{ M_0, M_1 \}$-secure, then
	\begin{equation}\label{eq:commute}
		\left\| M_1\circ M_0(\rho)- M_0\circ M_1(\rho)\right\|_1\le \varepsilon+\negl(n)
	\end{equation}
	for all efficiently preparable quantum states $\rho_{MR}$.
\end{theorem}
In terms of Defs. \ref{def:QSS-correct} and \ref{def:QSS-secure}, \expref{Equation}{eq:commute} must hold whenever $\{ M_0, M_1 \}\subset \mathcal N\cap\mathcal L$.

\begin{proof} The plan is as follows. The scheme is correct for $M_0$, the two-outcome measurement defined by a projector $P_0$. Using this, one can show that running the verification circuit (without discarding qubits, and delaying measurements), applying the reflection unitary $U=\one-2P_0$ and running the inverse of the verification circuit, produces a valid signed state. More precisely, if the described attack is applied to a signed state obtained from signing $\ket\psi$, the resulting state is a signed state that can be also obtained from signing $U\ket\psi$. But if $M_0$ and $M_1$ don't commute, then the unitary $U$ changes the outcome of $M_1$. Therefore the measurements have to commute, up to an error equal to the soundness error $\varepsilon$.
	
	We now explain the details. Assume that there exists efficiently preparable $\rho_{MR}$ such that 
	\begin{equation}\label{eq:assumption1}
	\left\| M_1\circ M_0(\rho)- M_0\circ M_1(\rho)\right\|_1>2\delta.
	\end{equation}
	Let $\ket{\psi}_{MRE}$ be an efficiently preparable purification of of $\rho_{MR}$. Let $P_i$ be the projector for the outcome $0$  of $M_i$, for $i=0,1$. The assumption \eqref{eq:assumption1} implies w.l.o.g. that 
	\begin{equation}\label{eq:noncommute-diff}
	\bigl|\left\|P_0P_1\ket\psi\right\|_2^2-\left\|P_1P_0\ket\psi\right\|_2^2\bigr|>\frac\delta 2.
	\end{equation}
	Define the reflection unitary $U_i=\one-2P_i$. We would like to show that an application of $U_0$ to $\ket\psi$ changes the outcome of $M_1$. We calculate
	\begin{align}\label{eq:assumption2}
	\bigl|\left\|P_1U_0\ket\psi\right\|_2^2-\left\|P_1\ket\psi\right\|_2^2\bigr|&=\bigl|\left\|P_1(\one-2P_0)\ket\psi\right\|_2^2-\left\|P_1\ket\psi\right\|_2^2\bigr|\\
	&=2\bigl|2\bra\psi P_0P_1P_0\ket\psi-\bra\psi P_1P_0\ket\psi-\bra\psi P_0P_1\ket\psi\bigr|.
	\end{align}
	Now we use Equation \eqref{eq:noncommute-diff} and rewrite the left hand side as
	\begin{align*}
	&\bigl|\left\|P_0P_1\ket\psi\right\|_2^2-\left\|P_1P_0\ket\psi\right\|_2^2\bigr|=\bigl|\bra\psi P_1P_0P_1\ket\psi-\bra\psi P_0P_1P_0\ket\psi\bigr|=\bigg|\bra\psi P_1P_0P_1\ket\psi\\
	&\!-\!\frac 1 2\bra\psi P_0P_1\ket\psi\!-\!\frac 1 2\bra\psi P_1P_0\ket\psi\!+\!\frac 1 2\bra\psi P_0P_1\ket\psi\!+\!\frac 1 2\bra\psi P_1P_0\ket\psi\!-\!\bra\psi P_0P_1P_0\ket\psi\bigg|\nonumber\\
	&\le \left|\bra\psi P_1P_0P_1\ket\psi-\frac 1 2\bra\psi P_0P_1\ket\psi-\frac 1 2\bra\psi P_1P_0\ket\psi\right|\nonumber\\
	&+\left|\frac 1 2\bra\psi P_0P_1\ket\psi+\frac 1 2\bra\psi P_1P_0\ket\psi-\bra\psi P_0P_1P_0\ket\psi\right|,
	\end{align*}
	where we have added a zero in the second equality and used the triangle inequality.
	In particular we obtain the inequality
	\begin{align}
	&\bigl|\left\|P_0P_1\nket\psi\right\|_2^2\!-\!\left\|P_1P_0\nket\psi\right\|_2^2\bigr|\!\le\! 2\max\!\bigg(\!\left|\nbra\psi P_1P_0P_1\nket\psi\!-\!\frac 1 2\nbra\psi P_0P_1\nket\psi\!-\!\frac 1 2\nbra\psi P_1P_0\nket\psi\right|,\nonumber\\
	&\left|\bra\psi P_0P_1P_0\ket\psi-\frac 1 2\bra\psi P_0P_1\ket\psi-\frac 1 2\bra\psi P_1P_0\ket\psi\right|\bigg)\label{eq:bound3}
	\end{align}
	Putting Equations \eqref{eq:noncommute-diff}, \eqref{eq:assumption2} as well as its counterpart with 0 and 1 interchanged, and \eqref{eq:bound3} together we arrive at the conclusion that either
	\begin{equation}\label{eq:U-matters}
	\bigl|\left\|P_1U_0\ket\psi\right\|_2^2-\left\|P_1\ket\psi\right\|_2^2\bigr|>\delta,
	\end{equation}
	or this equation holds with $0$ and $1$ interchanged. Assume w.l.o.g. that Equation \eqref{eq:U-matters} holds. Consider the adversary $\mathcal A$ that applies a Stinespring dilation of the verification algorithm, performs the unitary $U_0$ and then undoes the verification. This adversary produces a valid signed state by linearity.  Indeed, let $(\vk, \sk)$ be a given key pair, and $V_{\sk}$ and $W_{\vk}$ be Stinespring dilations of the signing and verification algorithms, respectively. Let furthermore $\tilde P_0=W_{\vk}^\dagger P_0 W_{\vk}$. By the correctness of the \QS for $M_0$, we have that $\bigl|\|\tilde P_0 V\ket\phi\|_2^2-\| P_0 \ket\phi\|_2^2\bigr|\le\negl(n)$, 
	for all $\ket\phi$, so in particular
	$\bigl|\|\tilde P_0 VP_0\ket\psi\|_2-\|V P_0 \ket\psi\|_2\bigr|\le\negl(n),$
	which implies $\|\tilde P_0 VP_0\ket\psi-V P_0 \ket\psi\|_2\le\negl(n)$  and similarly $\|\tilde P_0 V(\one-P_0)\ket\psi\|_2\le\negl(n)$.
	Together, these inequalities yield $\|\tilde P_0 V\ket\psi-V P_0 \ket\psi\|_2\le\negl(n)$
	and therefore 
	\begin{equation}\label{eq:U-cov}
	\|\tilde U_0 V\ket\psi-V U_0\ket\psi\|_2\le\negl(n),
	\end{equation}
	where $\tilde U_0=W_{\vk}^\dagger U_0 W_{\vk}$. This shows that the attack $U_0$ produces a valid ciphertext.
	Because of correctness of the \QS for $M_1$ we have that
	\begin{align}\label{eq:corrm1}
	\bigl|\|P_1WV U_0\ket\psi\|_2^2-\|P_1 U_0\ket\psi\|_2^2\bigr|\le\negl(n),
	\end{align}
	and the $\varepsilon$-one-time security implies that
	$\bigl|\|P_1W\tilde U_0V \ket\psi\|_2^2-\|P_1 \ket\psi\|_2^2\bigr|\le\varepsilon.$
	Combined with Equation \eqref{eq:U-cov} we therefore get
	\begin{align*}
	\bigl|\|P_1WV U_0\ket\psi\|_2^2-\|P_1 \ket\psi\|_2^2\bigr|\le\varepsilon+\negl(n)
	\end{align*}
	and hence, using Equation \eqref{eq:corrm1},
$
|\|P_1 U_0\ket\psi\|_2^2-\|P_1 \ket\psi\|_2^2|\le\varepsilon+\negl(n).
$	Via Equation \eqref{eq:U-matters}, this is a contradiction to assumption \eqref{eq:assumption1} as long as $\delta\ge\varepsilon+\negl(n)$, which finishes the proof.
\end{proof}

Note that for logarithmic-size message space, the conclusion of \expref{Theorem}{thm:imp-1} is equivalent to an upper bound on the commutator of the projectors forming the measurements $M_i$, as then all states are efficiently preparable. 

\expref{Theorem}{thm:imp-1} shows that simultaneously signing different properties of a quantum state is only possible if these properties are essentially classical, implying that the best possible security can be achieved by just measuring the $M_i$ sequentially (while at most incurring an error equal to the soundness parameter $\varepsilon$ plus a negligible function) and then signing the outcome classically.  This protocol would, however, destroy any quantum properties of the message state, even if no attack occurs! The following complementary impossibility result indicates that part of this loss of quantum information is unavoidable: If we require \emph{full correctness} for a \QS in the sense that the composition of $\Sign_{\sk}$ and $\Ver_{\vk}$ yields the identity channel on the plaintext space, no security can be achieved. The proof uses a similar idea as in \expref{Theorem}{thm:imp-1} (Stinespring, malleability unitary, un-Stinespring), details can be found in \expref{Section}{app:imposs2}.

\begin{theorem}\label{thm:imp-2}
	Let $\Pi = (\KeyGen, \Sign, \Ver)$ be a correct \QS, and let $M$ be a non-trivial two-outcome measurement. Then $\Pi$ is at most $(1-\negl(\secpar))$-one-time $\{M\}$-secure.
\end{theorem}

\subsubsection{Some variants}

Here we briefly discuss how the above impossibility results also extend to other, a priori plausible, definitions of quantum digital signatures. 

First, one might ask for a weak form of signatures where the receiver can {\em either} recover the message {\em or} check the authenticity (but not both.) Unfortunately, any useful formulation of this kind is likely to imply strong signatures via quantum error correction: first encode the message, then ``weak-sign'' the qubits; during verification, check a random subset of the qubits (at most half the code distance) and then use decoding to recover the message.

Second, one might try quantum keys. In \cite{Gottesman2001}, quantum keys are used to achieve information-theoretically secure signatures for classical messages; the public key is a quantum state, a copy of which can be requested by anybody. Doing this for quantum data must fail by the same arguments as in the proofs above: by correctness, an attack using the reflection corresponding to one of the signed measurements will implement the desired attack and leave the copy of the public key undisturbed. This implies impossibility even in a weaker model where the adversary has to hand back the public key. 
 
Finally, we might give the signer several copies of the message state, only one of which must be recovered by the verifier. However, no-cloning states that the ``valid/invalid'' bit output by verification cannot depend on the copy of the message recovered by the verifier. An adversary can thus run the same attack as in the proof of \expref{Theorem}{thm:imp-2}. After that, there is of course a mismatch between the plaintext message encoded in the signed message and any other parts that were created from the consumed message copies. As we observed, no-cloning tells us that this mismatch cannot be detected by the verification algorithm.

\section{Quantum Signcryption: Basic Ideas}

As we saw in 
\expref{Section}{sec:impossible}, 
signing states for public verification is impossible. We thus must resort to using secret verification. This might seem to imply that we have to use secret-key authentication schemes, which come with a dramatic cost: all pairs of parties wishing to communicate must share a key secret to that pair. As it turns out, there is a far better option: \emph{signcryption}~\cite{Zheng97,ADR02}.

In this section, in order to simplify the exposition, we start by addressing the basic scenario of {\em one-time security for the two-users setting}. We will expand to the more general case 
of {\em many-time security for the multi-user setting} 
in~\expref{Section}{sec:full}.

\subsection{Definition, Some Basic Facts, and a Construction}


\subsubsection{The basic notion}

Consider $m$ parties who require pairwise authenticated and encrypted communication. This can be done via symmetric-key authenticated encryption, at a cost of a quadratic number of key-exchange executions. Moreover, each party would need to store $m$ keys, and adding any new parties would require another round of $m$ key-exchange executions.

With signcryption, each party needs only to run key generation privately, publish their public key, and keep their private key. Sending a message is now a matter of ``signcrypting'' with the sender's private key (``sign'') and the receiver's public key (``encrypt.'') Receiving a message requires ``verified-decrypting" with the sender's public key (``verify'') and the receiver's private key (``decrypt.'') Adding new parties is trivial. 


\begin{defn}\label{def:QSC}
A {\em \textbf{quantum signcryption scheme}} (or $\QSC$) is a triple of QPT algorithms:
\begin{enumerate}
\item (key generation) $\KeyGen(1^n):$ output $(\sdk, \vek) \from \SDK \times \VEK$
\item (signcrypt) $\SigEnc: \SDK \times \VEK \times \states(\H_M) \rightarrow \states(\H_C)$
\item (verified decrypt) $\VerDec: \VEK \times \SDK \times \states(\H_C) \rightarrow \states(\H_M \oplus\ket{\bot})$
\end{enumerate}  
such that, for all $(\sdk_S, \vek_S), (\sdk_R, \vek_R) \from \KeyGen(1^n)$,
\begin{equation}\label{eq:qsc-correct}
\| \VerDec_{\vek_S, \sdk_R} \circ \SigEnc_{\sdk_S, \vek_R} - \id_M \oplus 0_\bot \|_\diamond \leq \negl(n)\,.
\end{equation}
\end{defn}

The key spaces $\SDK$ and $\VEK$ are classical and of size $\poly(n)$; $C$ and $M$ are quantum registers of at most $\poly(n)$ qubits. We have adopted the convention that, in the subscripts of $\SigEnc$ and $\VerDec$, the sender key always goes first. We will write $\SigEnc_{S, R}$ and $\VerDec_{S, R}$ to simplify this notation. As usual, we assume w.l.o.g. that $\sdk$ also includes $\vek$, and $\vek$ also includes \secpar.

A \QSC is used to transmit messages as follows. First, a sender $S$ selects a message (placing it in register $M$) and a receiver $R$. Then $S$ applies $\SigEnc_{S, R}$ to $M$, using their secret key $\sdk_S$ and the receiver's public key $\vek_R$. The resulting register $C$ is sent to $R$, who applies $\VerDec_{S, R}$ to $C$, using their secret key $\sdk_R$ and the sender's public key $\vek_S$. Correctness (i.e., \eqref{eq:qsc-correct}) requires that the overall channel implemented by this honest process should be $\id_M$ along with an ``accept'' output (indicated by ``$\oplus \ 0_\bot$''.)

\subsubsection{Signatures and encryption from signcryption}

Any quantum signcryption scheme trivially yields a \QS
scheme (\expref{Definition}{def:QSSnatural}), as follows.
\begin{prop}\label{prop:QSC-QS}
Let $\Pi = (\KeyGen, \SigEnc, \VerDec)$ be a signcryption scheme (\QSC.) Then the following is a signature scheme (\QS.)
\begin{itemize}
\item $\KeyGen'(1^n)$: $(\sdk_S, \vek_S) \from \KeyGen(1^n)$ and $(\sdk_R, \vek_R) \from \KeyGen(1^n)$; output signing key $\sk := (\sdk_S, \vek_R)$ and verification key $\vk := (\vek_S, \sdk_R)$.
\item $\Sign'_\sk := \SigEnc_\sk$ and $\Ver'_\vk := \VerDec_\vk$.
\end{itemize}
\end{prop}

Any \QSC also trivially yields a \PKQE 
scheme (\expref{Definition}{def:PKQE}).
\begin{prop}\label{prop:QSC-PKE}
Let $\Pi = (\KeyGen, \SigEnc, \VerDec)$ be a signcryption scheme (\QSC.) Then the following is a public-key encryption scheme (\PKQE.)
\begin{itemize}
\item $\KeyGen'(1^n)$: $(\sdk_S, \vek_S) \from \KeyGen(1^n)$ and $(\sdk_R, \vek_R) \from \KeyGen(1^n)$; output secret key $\dk := (\vek_S, \sdk_R)$ and public key $\ek := (\sdk_S, \vek_R)$.
\item $\Enc'_\ek := \SigEnc_\ek$ and $\Dec'_\dk := \VerDec_\dk$.
\end{itemize}
\end{prop}

\subsubsection{A basic construction}\label{sec:hybscdef}

We now define a generic ``hybrid'' construction of quantum signcryption 
in a similar way to~\expref{Section}{sec:hybqencdef}. 
We will use a classical signcryption scheme to signcrypt a random key, and then use that key to encrypt the quantum state with a (usually one-time secure) symmetric-key quantum encryption scheme.

\begin{construction}\label{cons:QS-generic}
Let $\Pi^\SC$ be a classical signcryption scheme, and $\Pi^\SKQE$ a symmetric-key quantum encryption scheme. 
Then 
we define 
the canonical hybrid 
quantum signcryption scheme $\PiHyb[\Pi^\SC, \Pi^\SKQE] = (\KeyGen, \SigEnc, \VerDec)$ as follows:
\begin{enumerate}
	\item $\KeyGen(1^n):$ output $(\sdk, \vek) \from \KeyGen^\SC(1^n)$;
	\item $\SigEnc_{S, R}:$ on input $\rho_M$, generate $k \from \KeyGen^\SKQE(1^n)$ and output\\ $(\SigEnc^\SC_{S, R}(k), \Enc^\SKQE_k(\rho))$;
	\item $\VerDec_{S, R}:$ on input $(c, \sigma_C)$, set $k =  \VerDec^\SC_{S, R}(c)$; if $k = \bot$, output $\bot$ and otherwise output $\Dec^\SKQE_k(\sigma)$.
\end{enumerate}
\end{construction}

A special case of this hybrid construction was proposed by Barnum et al.~\cite{BCG+02} (Section 5.1), but without any security definitions or proofs. We will later formally examine the security of this construction in various settings, and for various choices of $\Pi^\SC$ and $\Pi^\SKQE$.



\subsection{One-Time Signcryption Security}

We now consider the simplest nontrivial setting of quantum signcryption: a single sender $S$ 
sends 
a single signcrypted state to a single receiver $R$. Our goal here is to describe the most basic setting of interest, with definitions and security notions that require only minimal familiarity with previous work on quantum encryption and authentication. We will build up to the full setting in \expref{Section}{sec:full}.

\subsubsection{Outsider security}

In this setting, a third-party $\algo A$ attacks the channel between $S$ and $R$. First, $\algo A$ convinces $S$ to signcrypt any message; then $\algo A$ freely manipulates the resulting signcryption, finally convincing $R$ to unsigncrypt. 


\begin{experiment}\label{expr:one-time-outsider}
Let $\Pi = (\KeyGen, \SigEnc, \VerDec)$ be a \QSC. An \emph{\textbf{outsider attack}} (in the one-time, two-user setting) by a QPT $\algo A$ proceeds as follows. 
\begin{enumerate}
\item \emph{(Setup.)} Generate $(\sdk_S, \vek_S) \from \KeyGen(\secparam)$ and $(\sdk_R, \vek_R) \from \KeyGen(\secparam)$. Give $\vek_S$ and $\vek_R$ to $\algo A$.
\item \emph{(Signcrypt.)} $\algo A$ prepares a state $\rho_{MB}$ and $\SigEnc_{S, R}$ is applied to register $M$, yielding registers $C$ and $B$.
\item \emph{(Attack, unsigncrypt.)} $\algo A$ applies a channel $\Lambda_{CB \to CB}$ (possibly depending on $\vek_S, \vek_R$), and $\VerDec_{S, R}$ is applied to register $C$. 
\end{enumerate}
\end{experiment}
We refer to $\algo A$ as an \emph{\textbf{outside attacker}}, and define the (averaged) effective map
$$
\tilde \Lambda^{(\algo A, S, R)}_{MB \to MB} := \mathbb{E} \left[\VerDec_{S, R} \circ \Lambda \circ \SigEnc_{S, R}\right]\,
$$
where the expectation is over the keys and classical randomness. Our security definition will be based on
 (the computational version of)
\DNS security notion for 
one-time 
quantum authentication~\cite{DNS12,BW16}, adapted to our setting. Specifically, we will require that any (efficient) outside attack amounts to a trivial choice of effective map: implementing the ``identity map'' or the ``discard map'' on the plaintext space, with the latter resulting in a reject. 

\begin{defn}\label{def:one-time-outsider}
A \QSC $\Pi$ is (one-time, two-user) \textbf{outsider-secure} if, for every QPT outside attacker $\algo A$, there exists a QPT simulator $\Phi_{MB \to MB}$ of the form
$$
\Phi: \sigma_{MB} \longmapsto \Phi^\acc_{B\to B} (\sigma_{MB}) + \proj{\bot} \otimes \Phi^\crej_{B\to B} (\sigma_B)
$$
which is 
computationally
indistinguishable\footnote{Meaning that no QPT adversary can distinguish the two channels with more than negligible advantage over guessing at random.} 
from $\tilde \Lambda^{(\algo A, S, R)}$, and where $\Phi^\acc$ and $\Phi^\crej$ act on $B$ only.
\end{defn}

\paragraph{A strengthening}

Note that both $\SigEnc$ and $\VerDec$ receive, in addition to their input, a private parameter (i.e., $\sdk$) and a public parameter (i.e., $\vek$.) It is thus reasonable to define a stronger notion of security, where the adversary can set these parameters for the oracle calls occurring in the ``Signcrypt'' and ``Attack, unsigncrypt'' phases of \expref{Experiment}{expr:one-time-outsider}. Indeed, this models a very plausible real-world attack. Since our goal here is to develop the one-time setting as a minimal first stepping stone, we will examine this stronger security 
model 
later in \expref{Section}{sec:full}. 

\subsubsection{Insider security} 

An outside attacker knows all but two pieces of information: the private keys of $S$ and $R$. In ``insider security,'' we consider the case where one of these private keys is known to the attacker (if both are, security is impossible.) 

\paragraph{Insider security of sender}

First, consider the case of an adversary in possession of $\sdk_R$. The private information is now only $\sdk_S$, and the public information (i.e., the information accessible by the adversary) is $\sdk_R, \vek_R, \vek_S$. This means the adversary can unsigncrypt, so we can't expect $S$ to maintain any secrecy. Can $S$ expect unforgeability? If yes, then this would yield (as in \expref{Proposition}{prop:QSC-QS}) an unforgeable quantum signature scheme with secret signing key $(\sdk_S, \vek_R)$ and public verification key $(\vek_S, \sdk_R)$.) However, the impossibility results of \expref{Section}{sec:impossible} state that this induced signature scheme cannot satisfy even the most minimal security requirements. It follows that quantum signcryption cannot fulfill any insider security guarantees in this case.

Recall that \emph{non-repudiation} is a property of classical digital signatures. It means that, due to unforgeability, the recipient $R$ of a document $m$ signed by a sender $S$ can present the document and signature to a third party (e.g., the judge) as evidence that $m$ was indeed signed by $S$. If quantum signcryption is used instead, our impossibility results tell us that $R$ will be able to forge signcryptions; so, non-repudiation is impossible to achieve quantumly.


\paragraph{Insider security of receiver}

It remains to consider insider security where the adversary has $\sdk_S$. The private information is now only $\sdk_R$. By \expref{Proposition}{prop:QSC-PKE}, $\Pi$ is now a 
\PKQE 
scheme, with secret decryption key $(\sdk_R, \vek_S)$ and public encryption key $(\sdk_S, \vek_R)$.

\begin{defn}\label{def:one-time-insider}
A \QSC $\Pi$ is (one-time, two-user) \textbf{insider-secure} if the 
\PKQE 
scheme induced by making the sender key public (i.e., via \expref{Proposition}{prop:QSC-PKE}) satisfies \QINDCPA.
\end{defn}
As with outsider security, we will not (yet) concern ourselves with adversaries who can select the key parameters of the $\SigEnc$ and $\VerDec$ calls. 
This will be considered as part of the many-time security scenario in~\expref{Section}{sec:full}.

We remark that the above definitions of one-time security translate directly to the classical setting, by using standard notions of classical unforgeability and secrecy (see, e.g., \cite{KL14}) and plugging them into \expref{Definition}{def:one-time-outsider} and \expref{Definition}{def:one-time-insider}.

\subsubsection{Achieving one-time outsider and insider security} 

We now show that the hybrid scheme (with appropriate components) achieves one-time security.

\begin{theorem}\label{thm:one-time}
Let $\Pi^\SC$ be a classical signcryption scheme with one-time outsider and insider security. Let $\Pi^\SKQE$ be a \cDNS-secure \SKQE scheme (\expref{Definition}{def:DNS}). Then $\PiHyb[\Pi^\SC, \Pi^\SKQE]$ (\expref{Construction}{cons:QS-generic}) is a 
one-time outsider-secure and insider-secure \QSC.
\end{theorem}

\begin{proof}
We provide a brief sketch; for the full proof see \expref{Section}{app:one-time}. 

Insider security follows from the \QINDCPA security of the hybrid \PKQE construction (i.e. \expref{Theorem}{thm:PKE-hybrid-full}) and the fact that \cDNS implies \QIND \cite{BCG+02,AM17}. For outsider security, given an attack map $\Lambda$, we will construct a simulator as mandated by \expref{Definition}{def:one-time-outsider}, by using the \cDNS security of $\Pi^\SKQE$. First, if $\Lambda$ modifies the classical part  $c$ of the ciphertext, the decryption function will reject with overwhelming probability due to the unforgeability of $\Pi^\SC$. So we are free to assume $\Lambda$ does not change $c$. Thus we can replace $c$ by a fresh, unrelated $c'$ (by \INDCPA of $\Pi^\SC$) provided we swap $c$ and $c'$ back right before decryption. Since the quantum and classical parts of the ciphertext are now independent, we can view the adversary's channel on the quantum part as a \cDNS attack against $\Pi^\SKQE$. Since $\Pi^\SKQE$ is \cDNS secure, this attack has a simulator $\Phi$ which (effectively) ignores the quantum plaintext or discards it. This in turn yields a simulator $\Lambda'$ for $\Lambda$: apply $\Phi$ to the quantum part, and the identity to the classical part. The guarantees provided by $\cDNS$ on $\Phi$ ensure that $\Lambda'$ satisfies the conditions of \expref{Definition}{def:one-time-outsider}.
 \end{proof}


To get an explicit instantiation, we can let $\Pi^\SC$ be a classical signcryption scheme constructed, e.g., from Lamport signatures and standard post-quantum \INDCPA public-key \LWE encryption via ``encrypt-then-sign,'' and let $\Pi^\SKQE$ be the scheme $\Enc_k : \rho \mapsto C_k \left( \rho \otimes \proj{0^n} \right) C_k^\dagger$ where $\{C_k\}_k$ is the Clifford group. 


\section{Quantum Signcryption: Full Security}\label{sec:full}


We now describe how to ``upgrade'' from the (one-time, two-user) setting to the full  setting. This will involve three steps. First, in \expref{Section}{sec:many-time} we will describe how to upgrade to many-time security. The adversary will now have oracle access to $\SigEnc_{S, R}$ and $\VerDec_{S, R}$, but secrecy and authenticity must still be preserved. The usual classical security games for these notions do not make sense with quantum data. A recent approach of~\cite{AGM18} shows how to get around these issues for symmetric-key encryption; we will extend this approach to our setting. Second, in \expref{Section}{sec:multi-user} we will upgrade to multiple users. This means keeping track of ``IDs'' for each user, and adding some constraints, e.g., to prevent identity fraud. Here we will give a generic transformation (following~\cite{ADR02}) which turns any two-user secure \QSC into a multi-user secure \QSC simply by attaching IDs to the plaintext before signcrypting. Finally, in \expref{Section}{sec:ctxt} we will describe how to upgrade all of the above to ensure ciphertext authenticity. 

\subsection{Upgrading to the Many-Time Setting}\label{sec:many-time}


\subsubsection{Outsider security}

Here, a third-party adversary $\adver$ attacks a sender-receiver pair $(S, R).$ In addition to knowing the public information $\vek_S, \vek_R$, $\adver$ will now also be able to observe (or even control) some of the transmissions from $S$ to $R$. As usual, we model this by being as generous as possible to the adversary, giving them oracle access to both $\SigEnc_{S, R}$ and $\VerDec_{S, R}$. As signcryption provides secrecy, authenticity and integrity, it can be seen as an analogue of authenticated encryption. We will use the ideas of~\cite{AGM18} to define security, avoiding the usual issues with quantum no-cloning and measurement. For now, for ease of exposition, we only ask for plaintext integrity.

Let $\Pi = (\KeyGen, \SigEnc, \VerDec)$ be a quantum signcryption scheme with message register $M$ and signcryption register $C$. In addition, select a security parameter $n$ and an oracle QPT adversary $\adver$. Let $\ket{\phi^+}$ denote some choice of maximally entangled state, and $\Pi^+$ the projector $\proj{\phi^+}$ onto that state.

\begin{experiment}\label{exp:outsider-real}
The \textbf{real outsider experiment} $\OutReal(\Pi, \algo A, n)$:
	\begin{algorithmic}[1]
		\State $(\sdk_S, \vek_S) \from \KeyGen(1^n)$ and $(\sdk_R, \vek_R) \from \KeyGen(1^n)$;
		\State \textbf{output} $\algo A^{\SigEnc_{S, R}, \VerDec_{S, R}}(1^n)$.
	\end{algorithmic}
\end{experiment}


\begin{experiment}\label{exp:outsider-ideal}
The \textbf{ideal outsider experiment} $\OutIdeal(\Pi, \algo A, n)$:
	\begin{algorithmic}[1]
		\State $(\sdk_S, \vek_S) \from \KeyGen(1^n)$ and $(\sdk_R, \vek_R) \from \KeyGen(1^n)$;
		\State define channel $E_{M \rightarrow C}$:
			\begin{enumerate}[(1)]
			\item prepare $\ket{\phi^+}_{M'M''}$, store $(M'', M)$ in a set $\mathcal M$; \label{set-trap}
			\item apply $\SigEnc_{S, R}$ to $M'$; \textbf{return} result.
			\end{enumerate}
		\State define channel $D_{C \rightarrow M}$:
			\begin{enumerate}[(1)]
			\item apply $\VerDec_{S, R}$ to $C$, place results in $M'$; 
			\For {\textbf{each} $(M'', M) \in \mathcal M$}: \label{cheat-detect}
				\State apply $\{\Pi^+, \one - \Pi^+\}$ to $M'M''$; \textbf{if} {outcome is $0$}: \textbf{return} $M$;
			\EndFor
			\State	\textbf{return} $\egoketbra{\bot}$;
			\end{enumerate}
		\State \textbf{output} $\adver^{E,D}(1^n)$.
	\end{algorithmic}
\end{experiment}


\begin{defn}
A quantum signcryption scheme $\Pi$ is (many-time, two-user) \textbf{outsider secure} if for all QPT adversaries $\adver$,
 	\begin{equation}
 	\left|\Pr\left[\OutReal(\Pi, \algo A, n) \to \real \right] -
 	 	\Pr\left[\OutIdeal(\Pi, \algo A, n) \to \real \right]\right| \leq \negl(n).\nonumber
 	\end{equation}
\end{defn}

\subsubsection{Insider security}

As before, the remaining conditions to consider amount to allowing the adversary access to the private key of either $S$ or $R$ (but not both) which we can simply view as one of the two parties ($S$ or $R$) attacking the other. Also as before, the impossibility results for quantum signatures (\expref{Section}{sec:impossible}) imply that one cannot expect any security in the case where $R$ is the adversary. It remains to consider the case where $S$ is the adversary, and ask if the secrecy of $R$ can be preserved. Here we ask for the strongest notion of secrecy, which is \QINDCCAA.

\begin{defn}
A quantum signcryption scheme $\Pi$ is (many-time, two-user) \textbf{insider-secure}  if the public-key quantum encryption scheme induced by making the sender key public (i.e., via \expref{Proposition}{prop:QSC-PKE}) satisfies \QINDCCAA.
\end{defn}

As usual, we will first focus on the weaker, plaintext-only (\cDNS-style) \QINDwCCAA notion.

\subsubsection{Achieving many-time, two-user security}

It turns out that the hybrid 
signcryption 
scheme can achieve both outsider and insider (sender) security even in the many-time setting, given a sufficiently strong classical component. This is proven 
in~\expref{Section}{sec:many2hyb}.

\begin{theorem}\label{thm:many2hyb}
Let $\Pi^\SC$ be a many-time, two-user, outsider- and insider-secure signcryption. Let $\Pi^\SKQE$ be a \cDNS-secure symmetric-key quantum encryption. Then $\PiHyb[\Pi^\SC, \Pi^\SKQE]$ (\expref{Construction}{cons:QS-generic}) is a quantum signcryption scheme which is (many-time, two-user) outsider-secure and insider-secure.
\end{theorem}
\subsection{Upgrading to the Multi-User Setting}\label{sec:multi-user}

In the multi-user setting there are different users, each with unique ``IDs.'' We denote the ID of a party $P$ by $\ID_P \in \I$. We assume that there is an efficient public lookup for the map $\ID_P \mapsto \vek_P$ (this is usually achieved by a root of trust, PKI, or similar). In this setting we need to change the syntax of signcryption in order to account for the adversary's ability to {\em spoof sender's and/or receiver's identities}. In order to protect the scheme against these attacks we need to bind the users' identities to the corresponding keys, and impose extra correctness conditions based on the match between identities and keys. There are various ways to achieve this; we opt for the simplified notation below. In particular, we assume that every user $P$ has a single keypair (it is possible to remove this restriction, e.g., by adding a {\em key index} to the output of \KeyGen and managing the index together with the identities.)

\begin{defn}\label{def:QSCmulti}
A {\em (multi-user) quantum signcryption scheme (or $\muQSC$)} is a triple of QPT algorithms:
\begin{itemize}
\item (key generation) $\KeyGen:(\secparam, \ID_P \in \I) \mapsto (\sdk_P, \vek_P) \in \SDK \times \VEK$
\item (signcrypt) $\SigEnc: \SDK \times \VEK \times \states(\H_M) \rightarrow \states(\H_C)$
\item (unsigncrypt) $\VerDec: \SDK \times \VEK  \times \states(\H_C) \rightarrow  \I^2 \otimes \states(\H_M \oplus\ket{\bot})$
\end{itemize}  
such that:
\begin{enumerate}
\item $\| \VerDec_{\sdk_R, \vek_S} \circ \SigEnc_{\sdk_S, \vek_R} -  \{S,R\}\otimes\left( \id_M \oplus 0_\bot \right) \|_\diamond \leq \negl(n)$
\item For $\VerDec_{\sdk_R, \vek_S}(X) = (S',R',Y)$, if $Y\neq \bot$, then $(R,S)=(R',S')$
\end{enumerate}
for all $(\sdk_P, \vek_P) \in \emph{\supp} \KeyGen(\secparam, \ID_P)$, and where we assume w.l.o.g. that $\sdk_P$ also includes $\vek_P$, and $\vek_P$ also includes $\secparam$ and $\ID_P$, for every $(\sdk_P, \vek_P) \from \KeyGen(\secparam, \ID_P), \forall \ P \in \I$. 
\end{defn}


The output of $\VerDec_{\sdk_R, \vek_S}$ includes user identities. Correctness demands that such identities are the ones corresponding to the keys used in \SigEnc. In theory nothing forbids a malicious adversary from modifying a ciphertext such that the resulting identities output by \VerDec are different. In this sense, an ``identity fraud attack'' for signcryption in the multi-user setting results in a forgery attack.

\subsubsection{Outsider security}
The scenario here is similar to outsider security for the two-user setting: an external adversary \adver mounts an attack against a sender/receiver pair $(S,R)$. However, there are two fundamental differences: (1.) there are many users in \I, and hence many possible $(S,R)$ pairs; and (2.) beyond the ``usual" attack scenarios covered by the two-user case, \adver might be able to attack the signcryption scheme by performing {\em identity fraud} (i.e., spoofing the identity of the sender or the receiver of a signcrypted message). Dealing with the latter is easy: as in \expref{Definition}{def:QSCmulti}, in order to mount an identity fraud attack \adver has to modify the identities output by the $\VerDec$ algorithm. This means that such attacks are actually a special case of {\em message forgeries}, which are already covered by the real vs. ideal approach in the outsider security scenario of \expref{Section}{sec:many-time}.

The other issue is more subtle. There might be insecure schemes which produce weaker keys for certain user $\ID$s, or there might be schemes where the compromise of a certain number of users (and hence knowledge of their secret keys) also compromises the security of other users. In order to deal with all these scenarios, we will adopt a very conservative approach: (1.) the adversary is given the possibility of choosing $S$ and $R$ among \I; and (2.) the adversary receives the secret keys of {\em all the other users}.

We define the modified experiments \MOutReal and \MOutIdeal accordingly, 
and security is given as usual in terms of indistinguishability of the two.

\begin{experiment}\label{exp:multi-outsider-real}
	The \textbf{multi-user real outsider experiment} $\MOutReal(\Pi, \adver, \secpar)$:
	\begin{algorithmic}[1]
		\State \textbf{for} every $\ID_P \in \I$ \textbf{do:} $(\sdk_P, \vek_P) \from \KeyGen(\secparam, \ID_P)$.
		\State \adver gets as input \secparam, \I, and the list of all public keys $\left\{\vek_P\right\}$ for all $\ID_P \in \I$;
		\State \adver outputs two identities $S,R \in \I$;
		\State \adver receives the list of all secret keys $\left\{\sdk_P\right\}$ for all $\ID_P \in \I \setminus \left\{ S,R\right\}$;
		\State \textbf{return} the output of $\adver^{\SigEnc_{S, R}, \VerDec_{S, R}}$.
	\end{algorithmic}
\end{experiment}

\begin{experiment}\label{exp:multi-outsider-ideal}
	The \textbf{multi-user ideal outsider experiment} $\MOutIdeal(\Pi, \adver, \secpar)$:
	\begin{algorithmic}[1]
		\State \textbf{for} every $\ID_P \in \I$ \textbf{do:} $(\sdk_P, \vek_P) \from \KeyGen(\secparam, \ID_P)$.
		\State \adver gets as input \secparam, \I, and the list of all public keys $\left\{\vek_P\right\}$ for all $\ID_P \in \I$;
		\State \adver outputs two identities $S,R \in \I$;
		\State \adver receives the list of all secret keys $\left\{\sdk_P\right\}$ for all $\ID_P \in \I \setminus \left\{ S,R\right\}$;
		\State define channel $E_{M \rightarrow C}$:
		\begin{enumerate}[(1)]
			\item prepare $\ket{\phi^+}_{M'M''}$, store $(M'', M)$ in a set $\mathcal M$;
			\item apply $\SigEnc_{S, R}$ to $M'$; \textbf{return} result.
		\end{enumerate}
		\State define channel $D_{C \rightarrow M}$:
		\begin{enumerate}[(1)]
			\item apply $\VerDec_{S, R}$ to $C$, place results in $M'$; 
			\For {\textbf{each} $(M'', M) \in \mathcal M$}: 
			\State apply $\{\Pi^+, \one - \Pi^+\}$ to $M'M''$;
			\State \textbf{if} {outcome is $0$}: \textbf{return} $M$;
			\EndFor
			\State	\textbf{return} $\egoketbra{\bot}$;
		\end{enumerate}
		\State \textbf{return} the output of $\adver^{E,D}$.
	\end{algorithmic}
\end{experiment}

\begin{defn}
A multi-user \QSC 
scheme $\Pi$ is \textbf{outsider secure} if for all QPT $\adver$,
 	\begin{equation}
 	\left|\Pr\left[\MOutReal(\Pi, \adver, \secpar) \to \real \right] -
 	 	\Pr\left[\MOutIdeal(\Pi, \adver, \secpar) \to \real \right]\right| \leq \negl(\secpar).\nonumber
 	\end{equation}
\end{defn}

\subsubsection{Insider security}

In addition to picking a sender and receiver of his choice, the adversary is now allowed to access the private key of either of the two (but not both). Also as before, given the impossibility results for quantum signatures (\expref{Section}{sec:impossible}), it makes sense to consider the case where $S$ is the adversary, and ask if the secrecy of $R$ can be preserved. Again, we ask for the strongest notion of secrecy. However, given the multi-user setting, we must now consider a multi-user version of \QINDCCAA, where the adversary gets to pick the target $R$ to attack, and receives all other user's secret keys. We omit the details here, as the resulting definition follows a similar approach as in the outsider security case. 

\begin{defn}
A multi-user \QSC $\Pi$ is \textbf{insider-secure} if the public-key quantum encryption scheme induced by making the sender key public (i.e., via \expref{Proposition}{prop:QSC-PKE}) satisfies (the multi-user version of) \QINDCCAA .
\end{defn}

\subsubsection{Achieving many-time, multi-user security}

The multi-user definitions can be fulfilled in a way analogous to the classical case. Here \cite{ADR02} describe a simple generic transformation from a two-user secure scheme to a multi-user secure scheme, which proceeds as follows. Whenever signcrypting any plaintext, we first attach to the plaintext the IDs of both the sender and the intended receiver. When applying verified decryption, we check whether the IDs attached to the plaintext are correct. If so, we output the plaintext, otherwise reject. Arguing multi-user security of this transformation then reduces in a straightforward way to two-user security, essentially as done in \cite{ADR02}. 
We omit a full proof of security for this construction, and leave it for a future work. However, the details should be easy to infer from the previous discussion and follow the same general strategy.

\subsection{Upgrading to Ciphertext Authentication}\label{sec:ctxt}
As discussed above, the security notions we have developed for signcryption provide for \emph{plaintext security.} For instance, in the outsider security case, this means that the adversary cannot perform any attack which modifies the underlying plaintext. The advantage of this approach is that the security games and definitions are rather simple to state and describe, and intuitive. However, this leaves open the possibility of adversaries that modify ciphertexts without being noticed. As it turns out, we can address this case as well, roughly following the route followed by~\cite{AGM18} in the private-key encryption setting. Here we only briefly describe the modifications, as they are essentially identical to those in~\cite{AGM18}.

First, we observe that the characterization lemma (\expref{Lemma}{lem:PKE-char}) shows that all quantum encryption schemes have a simple ``attach auxiliary state, then apply keyed isometry $V_k$'' form. This form is also efficient in all schemes we are aware of (we state this formally as \expref{Condition}{con:efficient}), but may \emph{in principle} be inefficient.

Next, with this characterization in hand, we can replace the ``trap setting'' step of each relevant security game (e.g., \expref{Step}{set-trap} in the \OutIdeal experiment) as follows. We still encrypt half of a maximally-entangled state $\phi^+$ and store the input and half of $\phi^+$. In addition, we now also store the classical randomness used to sample the aforementioned auxiliary state (prior to applying $V_k$.)

We then correspondingly adjust the second ``cheat detection'' step of each relevant security game (e.g., \expref{Step}{cheat-detect} in the \OutIdeal experiment), as follows. Given a ciphertext, we first undo the isometry $V_k^\dagger$. We then check if the auxiliary state agrees with the stored randomness. If it does, we additionally perform the entanglement check as before. If both checks say ``yes,'' we supply the stored state. Otherwise we output $\bot$. (In the insider-security case, it's slightly different: ``yes'' now means that the adversary is attempting to decrypt the challenge, so we simply terminate and output ``cheat.'')

This transformation yields enhanced cheat-detection games for both outsider-security (analogous to \QAE) and insider-security (analogous to \QINDCCAA.) Security is then still defined in terms of the advantage of adversaries at distinguishing the test game from the corresponding cheat-detection game for outsider security, and the winning advantage in the test game over the cheat-detection game for insider security, respectively. Provided that the relevant construction now uses \QCA \cite{AGM18} instead of \cDNS quantum authentication as a building block, the security proofs carry over essentially unchanged.

%
\section{Acknowledgments}\label{sec:ack}
%
%

The authors would like to thank Yfke Dulek, Christopher Portmann, Christian Schaffner, and Yi-Kai Liu for helpful feedback about this work. G.A. and C.M. thank Rolando La Placa for helpful discussions. Part of this work was done while T.G. was supported by IBM Research Zurich. T.G. acknowledges financial support from the H2020 projects FutureTPM (agreement 779391) and  FENTEC (agreement 780108). GA acknowledges support from the NSF under grant CCF-1763736, from the U.S. Army Research Office under Grant Number W911NF-20-1-0015, and from the U.S. Department of Energy under Award Number DE-SC0020312. CM was supported by a NWO VENI grant (Project No.VI.Veni.192.159).




\appendix
\section{Appendix}\label{app:security}

In the appendix, ``indistinguishable'' always means computationally indistinguishable.


\subsection{Detailed Proof of \expref{Theorem}{thm:QAEhybrid}}\label{sec:QAE-secu}

Here, we provide a few more details for the proof of~\expref{Theorem}{thm:QAEhybrid}. 
We only prove this theorem for the weak, plaintext-based version \wQAE starting from a \cDNS \SKQE scheme and a \CAE \SKE scheme. The general result follows from the discussion in \expref{Section}{sec:ctxt}. 
We have to show that the games \QAEideal and \wQAEreal using the scheme $\PiKEM[\Pi^\Cl, \Pi^\Qu]$ are indistinguishable for any QPT adversary \adver.

First of all we prove the
following lemma, 
which is an easy consequence of \DNS security.

\begin{lem}\label{lem:multiDNS}
	Let $\Pi = (\KeyGen, \Enc, \Dec)$ be a \DNS (\cDNS)-secure $\SKQE$, and let $\Lambda_{CB\to C^\ell B}$ be a CP map. Then there exist CP-maps $\Lambda^{(i)}_{B\to \tilde B}$, $i=1,...,\ell$ and $\Lambda^{\crej}_{B\to \tilde B}$ that sum to a TP map, such that:
	\begin{align*}
		&\Bigg\| \left(\sum_{i=1}^\ell \id_{M\to M_i}\otimes \left(\bot^{\otimes (\ell-1)}\right)_{M_{i^c}}\otimes \Lambda^{(i)}_{B\to B}+\left(\bot^{\otimes \ell}\right)_{M^\ell}\otimes \Lambda^\bot_{B\to B}\right)  \\
		&-
		\Lambda^\Pi_{MB\to M^\ell\tilde B} \Bigg\|_\diamond\le\negl(n),
	\end{align*}
	(or the computational 
	version of the above), where 
	$$
	\Lambda^\Pi_{MB\to M^\ell\tilde B}=\mathbb{E}\left[\Dec_k^{\otimes \ell}\circ\Lambda_{CB\to C^\ell B}\circ \Enc_k\right].
	$$
\end{lem}

\begin{proof}
	\DNS-security implies that for each $i$ there exist $\Lambda^{\acc,i}_{B\to \tilde B}$ and $\Lambda^{\crej,i}_{B\to \tilde B}$
	\begin{equation}\label{eq:equation1337}
		\left\|\tr_{M_{i^c}}\circ\Lambda^\Pi_{MB\to M^\ell\tilde B}-\id_{M\to M_i}\otimes\Lambda^{\acc,i}_{B\to \tilde B}+\proj\bot_{M_i}\otimes\Lambda^{\crej,i}_{B\to \tilde B}\right\|_\diamond\le\negl(n).
	\end{equation}
	Consider 
	the two maps $\Lambda^\Pi_{MB\to M^\ell\tilde B}-\proj\bot_{M_1}\otimes\bra\bot_{M_i}\Lambda^\Pi_{MB\to M^\ell\tilde B}\ket{\bot}_{M_1}$ 
	and  $\bra\bot_{M_1}\Lambda^\Pi_{MB\to M^\ell\tilde B}\ket{\bot}_{M_1}$ 
	separately. The first map is clearly close to $\id_{M\to M_1}\otimes \left(\bot^{\otimes (\ell-1)}\right)_{M_{1^c}}\otimes \Lambda^{\acc,i}_{B\to \tilde B}$, as the outputs on $M_i$, $i\neq 1$ do not depend on the input, and hence have to be equal to $\bot$ according  to Equation \eqref{eq:equation1337}. We now apply apply the same argument for $\bra\bot_{M_1}\Lambda^\Pi_{MB\to M^\ell\tilde B}\ket{\bot}_{M_1}$ and continue iteratively, which proves the lemma. The computational case works similarly.
\end{proof}

We begin 
the proof of~\expref{Theorem}{thm:QAEhybrid} 
by defining a hybrid game in the following way: 

\begin{experiment}
	The experiment \textsf{Hybrid $0$}:
	\begin{algorithmic}[1]
		\State $k \from \KeyGen(\secparam)$ ;
		\State define channel $E_{M \rightarrow C}$ as follows:
		\begin{enumerate}[(1)]
			\item run $\Enc^{\mathsf{Hyb}}_k$ on $M$, put result in registers $C^\Qu {C'}^\Cl$;
			\item sample $k''\from\KeyGen^{\Qu}$ place $\Enc^\Cl_k(k'')$ in $C^\Cl$;
			\item store $({C'}^\Cl,k'')$ in a set $\mathcal S$;
			\item \textbf{return} $(C^\Qu, C^\Cl)$.
		\end{enumerate}
		\State define channel $D_{C \rightarrow M}$ as follows:
		\begin{enumerate}[(1)]
			\For {\textbf{each} $({C'}^\Cl,{C''}^\Cl) \in \mathcal S$}: 
			\State compute $\tilde k$ by running $\Dec^\Cl_k$ on $C^\Cl$;
			\If {$\tilde k=k''_j$:}
			\State apply $\Dec^{\mathsf{Hyb}}$ to $(C^\Qu, {C'}^\Cl)$, place result in $M$ \textbf{output} $M$;
			\EndIf
			\EndFor
			\State	\textbf{return} $\egoketbra{\bot}$.
		\end{enumerate}
		\State \textbf{output} $\adver^{E,D}(\secparam)$.
	\end{algorithmic}
\end{experiment}

Suppose there exists an adversary \adver that can distinguish \QAEreal and \textsf{Hybrid $0$}. Then we can build an adversary \adver' that distinguishes the real and ideal worlds, \CAEreal and \CAEideal, in the real vs ideal characterization of \CAE by Shrimpton \cite{Shrimpton04}. \adver' runs \adver, answering its queries in the following way. On an encryption query with register $M$, sample $k'\from\KeyGen^{\Qu}$ and encrypt $M$ with $\Enc_{k'}^\Qu$ to obtain $C^\Qu$. Now send $k'$ to the encryption oracle and put the result into $C^\Cl$. Return $(C^\Qu, C^\Cl)$. For a decryption query with register $C=(C^{\Qu}, C^\Cl)$, decrypt the contents of $C^\Cl$ using the decryption oracle and use the result $k'$ to decrypt the contents of $C^\Qu$ into $M$. Return $M$. It is easy to see that \adver' uses the experiment \CAEreal for $\Pi^\Cl$ to make \adver play \QAEreal, and   the experiment \CAEideal for $\Pi^\Cl$ to make \adver play \textsf{Hybrid $0$}. Hence the indistinguishability between \QAEreal and \textsf{Hybrid $0$} follows from the \CAE security of $\Pi^\Cl$.

We continue to show that the experiments \textsf{Hybrid $0$} and \wQAEideal are computationally indistinguishable due to the \cDNS security of $\Pi^\Qu$. First observe that $\Pi^\Cl$ is in particular \INDCPA secure, and therefore randomized. More precisely it holds that for $c\from \Enc_k^{\Cl}$ and $c'\from \Enc_k^{\Cl}$, the probability of $c=c'$ is negligible. Let us therefore assume that the ciphertexts in $\mathcal S$ in \textsf{Hybrid $0$} are all distinct. It follows that for each decryption oracle call, there is at most one entry in $({C'}^\Cl,{C''}^\Cl) \in\mathcal S$ such that $C^\Cl={C''}^\Cl$. Let \adver be an adversary against the \wQAE security of $\PiKEM$, and assume that it makes $q$ encryption queries. We build $q$ hybrid experiments as follows:
\begin{experiment}
	The experiment \textsf{Hybrid $i$}:
	\begin{algorithmic}[1]
		\State $k \from \KeyGen(\secparam)$ ;
		\State define a stateful channel $E_{M \rightarrow C}$ as follows:
		\begin{enumerate}[(1)]
			\item let $j-1$ be the number of times $E_{M \rightarrow C}$ has been called before;
			\item \textbf{If} $j\le i$ move content of $M$ to $M'_j$, prepare $\ket{\phi^+}_{M_jM''_j}$, store $(M'_j, M''_j)$ in $\mathcal M$;
			\item run $\Enc^{\mathsf{Hyb}}_k$ on $M$, put result in registers $C^\Qu {C'}_j^\Cl$;
			\item sample $k''_j\from\KeyGen^{\Qu}$ place $\Enc^\Cl_k(k''_j)$ in $C^\Cl$;
			\item store $({C'}_j^\Cl,k''_j)$ in a set $\mathcal S$;
			\item \textbf{return} $(C^\Qu, C^\Cl)$.
		\end{enumerate}
		\State define channel $D_{C \rightarrow M}$ as follows:
		\begin{enumerate}[(1)]
			\For {\textbf{each} $({C'}_i^\Cl,{C''}_i^\Cl) \in \mathcal S$}: 
			\State compute $\tilde k$ by running $\Dec^\Cl_k$ on $C^\Cl$;
			\If {$\tilde k=k''_j$:}
			\State apply $\Dec^{\mathsf{Hyb}}$ to $(C^\Qu, {C'}^\Cl)$, place result in $M$ ;
			\If {$j\le i$:}
			\State measure $\proj{\phi^+}$ vs. $\one-\proj{\phi^+}$ on $MM''_j$; \textbf{if} outcome is $0$ swap $M$ and $M'_j$, \textbf{else} prepare $\proj\bot_M$;
			\EndIf
			\State \textbf{output} $M$.
			\EndIf
			\EndFor
			\State	\textbf{return} $\egoketbra{\bot}$.
		\end{enumerate}
		\State \textbf{output} $\adver^{E,D}(\secparam)$.
	\end{algorithmic}
\end{experiment}

We will now use \cDNS security of $\Pi^\Qu$. Applying the definition of \cDNS to the outputs of an attack map $\Lambda_{C\to C^\ell}$ sequentially shows that \cDNS security also implies the existence of a simulator in the case of one encryption, but many decryptions. Here, the simulator has the form
\begin{equation}
	\Lambda^{\mathsf{Sim}}_{M\to M^\ell}=\sum_{i=1}^\ell \id_{M\to M_i}\otimes \left(\bot^{\otimes (\ell-1)}\right)_{M_{i^c}}\otimes \Lambda^{(i)}_{B\to B}+\left(\bot^{\otimes \ell}\right)_{M^\ell}\otimes \Lambda^\bot_{B\to B},
\end{equation}
where $B$ is the adversary's side information register. We can replace the combination "$i+1$-st encryption - attack - decryptions matching $k''_{i+1}$" by this simulator in both \textsf{Hybrids $i$} and $i+1$, yielding \textsf{Hybrids $i_a$} and $(i+1)_b$ which are indistinguishable from \textsf{Hybrids $i$} and $i+1$, respectively. It remains to observe that the entanglement tests in \textsf{Hybrid} $i+1$ exactly project onto the respective identity contribution of $\Lambda^{\mathsf{Sim}}$ (a similar idea to the Broadbent-Waynewright simulator \cite{BW16}), i.e.
\begin{equation}
	\bra{\phi^+}_{MM_i}\Lambda^{\mathsf{Sim}}_{M\to M^\ell}\left(\phi^+_{MM_i}\otimes (\cdot)_B\right)\ket{\phi^+}_{MM_i}=\left(\bot^{\otimes (\ell-1)}\right)_{M_{i^c}}\otimes \Lambda^{(i)}_{B\to B}.
\end{equation}
This immediately implies that \textsf{Hybrids $i_a$} and $(i+1)_b$ are identical. We conclude that \textsf{Hybrid $0$} and \textsf{Hybrid $q$} are indistinguishable. It remains to show that \textsf{Hybrid $q$} and \wQAEideal are indistinguishable. This follows by an application of the Gentle Measurement Lemma \cite{Winter99}: while \textsf{Hybrid $q$} first checks which register $M''$ should be used together with the decrypted quantum plaintext to measure $\proj{\phi^+}$ vs. $\one-\proj{\phi^+}$ by means of comparing $\tilde k$ and the $k''$s, \wQAEideal just tries them one by one. For the ones that do not fit, though, this measurement yields "not maximally entangled" with overwhelming probability, and the gentle measurement lemma implies that the decrypted plaintext is not disturbed.

\subsection{Detailed Proof of~\expref{Theorem}{thm:KEMwCCA2}}\label{sec:proofwCCA2}

Let $\adver = (\adver_1, \adver_2)$ be an adversary playing \QINDCCAAtest. We will go through a game hopping, where the winning probability of the adversary always increases (or decreases at most negligibly), up to an experiment which is equivalent to \QINDwCCAAfake. First of all, we restate the two experiments explicitly.

\begin{experiment}\label{exp:INDwCCAA-test-supp}
	The $\QINDCCAAtest(\Pi, \adver, \secpar)$ experiment:
	\begin{algorithmic}[1]
		\State $\chall$ runs $(\dk,\ek) \from \KeyGen(\secparam)$ and flips a coin $b \inrand \bit$;
		\State $\algo A_1$ receives $\ek$ and oracle access to $\Dec_\dk$;
		\State $\algo A_1$ prepares a side register $S$, and sends to $\chall$ a challenge register $M$;
		\State $\chall$ puts into $C$ either $\Enc_\ek(M)$ (if $b=0$) or $\Enc_\ek(\tau_M)$ (if $b=1$);
		\State $\algo A_2$ receives registers $C$ and $S$ and oracle access to $\Dec_\dk$;
		\State $\algo A_2$ outputs a bit $b'$.  \textbf{If} $b'=b$, \textbf{output} \win; otherwise \textbf{output} \rej.
	\end{algorithmic}
\end{experiment}

\begin{experiment}\label{exp:INDwCCAA-fake}
	The $\QINDwCCAAfake(\Pi, \adver,\secpar)$ experiment:
	\begin{algorithmic}[1]
		\State $\chall$ runs $(\dk,\ek) \from \KeyGen(\secparam)$;
		\State $\algo A_1$ receives $\ek$ and oracle access to $\Dec_\dk$;
		\State $\algo A_1$ prepares a side register $S$, and sends to $\chall$ a challenge register $M$;
		\State $\chall$ discards $M$, prepares $\ket{\phi^+}_{M'M''}$, and stores $M''$; then $\chall$ encrypts $M'$  (using $\ek$) and puts the resulting ciphertext into $C'$;
		\State $\algo A_2$ receives registers $C'$ and $S$ and oracle access to $D_\dk$, where $D_\dk$ is defined as follows. On input a register $C$:
		\begin{enumerate}[(1)]
			\State $\chall$ applies $\Dec_\dk$ to $C$, places result in $M$; 
			\State $\chall$ applies $\{\Pi^+, \one - \Pi^+\}$ to $M M''$;
			\State \algorithmicif\  {the outcome is $1$} \algorithmicthen: abort and \textbf{output} \textsf{cheat}; \algorithmicelse \ \textbf{return} $M$;
		\end{enumerate}
		\State $\chall$ draws a bit $b$ at random. \textbf{If} $b=1$, \textbf{output} \cheat; if $b=0$ \textbf{output} \rej.
	\end{algorithmic}
\end{experiment}

We start now with the game-hopping.

\Game{0:} this is just \QINDCCAAtest.

\Game{1:} as \game{0}, except that \chall will prepare two entangled registers $\ket{\phi^+}_{M'M''}$; then if $b = 1$ instead of encrypting the maximally mixed state $\tau_M$, he will discard $M$ and return to \adver the encryption of $M'$ instead. Clearly the winning probability of \adver is unaffected by this, otherwise \adver could distinguish $\tau_M$ from $\phi^{+}_{M'}$ without access to $\phi^{+}_{M''}$.

\Game{2:} as \game{1}, but instead of oracle access to $\Dec_\dk$, $\adver_2$ gets access to a modified oracle $D'_\dk$ which, on input a ciphertext register $C$, does the following:

\begin{enumerate}
	\item applies $\Dec_\dk$ to $C$, places result in $M$; 
	\item applies $\{\Pi^+, \one - \Pi^+\}$ to $M M''$;
	\item \algorithmicif\  {the outcome is $1$} \algorithmicthen: abort game and \textbf{output} \textsf{\win}; \algorithmicelse \ \textbf{return} $M$;
\end{enumerate}

Notice that $D'_\dk$ is the same as $D_\dk$ in \QINDwCCAAfake, except it aborts with \win instead of \cheat whenever it detects a challenge plaintext replay. Clearly, the probability of \adver of winning this game does not decrease in respect to the previous game, hence so far we have:
\begin{equation*}
	\Pr \left[ \adver \ \win\mbox{s} \ \game{2} \right] \geq \Pr \left[ \adver \ \win\mbox{s} \ \QINDCCAAtest \right].
\end{equation*}

\Game{3:} as \game{2}, but \chall ``blinds" the classical part of the challenge ciphertext, replacing it with a different one, and $D'_\dk$ is modified in order to undo the blinding whenever queried on the fake ciphertext, making the substitution transparent to the adversary. The experiment looks as follows.

\begin{experiment}\label{exp:game3CCA2}
	The $\game{3}(\Pi, \adver, \secpar)$ experiment:
	\begin{algorithmic}[1]
		\State $\chall$ runs $(\dk,\ek) \from \KeyGen(\secparam)$, prepares maximally entangled registers $\ket{\phi^+}_{M'M''}$, and flips a bit $b \rand \bin$;
		\State $\algo A_1$ receives $\ek$ and oracle access to $\Dec_\dk$;
		\State $\algo A_1$ prepares a side register $S$, and sends to $\chall$ a challenge register $M$;
		\State \algorithmicif\ $b = 1$, \algorithmicthen: \chall discards $M$ and replaces it with $M'$;
		\State $\chall$ encrypts $M$ (using $\ek$), puts the resulting ciphertext into $C$, and records the classical part (PKE ciphertext) $t$ of the resulting ciphertext;
		\State \chall generates a fresh one-time key $k' \from \KeyGen^\SKQE(\secparam)$;
		\State \chall encrypts $t' \from \Enc^\PKE_\ek(k')$, records $t'$, and replaces the classical $t$ with $t'$ in the ciphertext register $C$;
		\State $\algo A_2$ receives registers $C$ and $S$, and oracle access to $D^t_\dk$ defined as follows. On input a register $C$:
		\begin{enumerate}[(1)]
			\State measure the classical-subsystem part of $C$: \algorithmicif\  {the outcome is $t'$} \algorithmicthen: replace the classical-subsystem part of $C$ with $\egoketbra{t}$;
			\State applies $\Dec_\dk$ to $C$, places result in $M$; 
			\State applies $\{\Pi^+, \one - \Pi^+\}$ to $M M''$;
			\State \algorithmicif\  {the outcome is $1$} \algorithmicthen: abort game and \textbf{output} \textsf{\win}; \algorithmicelse \ \textbf{return} $M$;
		\end{enumerate}
		\State $\algo A_2$ outputs a bit $b'$.  \textbf{If} $b'=b$, \textbf{output} \textsf{win}; otherwise \textbf{output} \textsf{fail}.
	\end{algorithmic}
\end{experiment}

Notice that, since in the algorithm $\Dec$ of $\PiKEM$ the first step is to measure the classical part of the register $C$, the new measurement introduced by $D^t_\dk$ does not disturb \adver's behavior. Moreover, the substitution of $t$ with $t'$ is undetectable for \adver, otherwise we could build a reduction \badver against the \INDCCAA security of $\Pi^\PKE$, in the following way:

\begin{enumerate}
	\item \badver plays the \INDCCAA game against $\chall^\CCAA$ (for a secret bit $b$), simulating \chall for \adver in the obvious way (forwarding \ek, simulating \Dec by decrypting classical keys with his own $\Dec^\PKE_\dk$ oracle and then doing the \SKQE decryption himself).
	\item During the challenge from \adver, \badver generates a fresh $k$ and uses it as a challenge to $\chall^\CCAA$, receiving back a ciphertext $t$;
	\item Flip a coin $b' \rand \bin$ to decide whether for the rest of the game \badver will try to simulate \game{2} ($b' = 0$) or \game{3} ($b' = 1$) for \adver. In the first case just return to \adver the correctly formed ciphertext register $C$, and simulate decryption queries on $t$-parts by defining their decryption as $k$. In the latter case instead, generate a fake encryption $t'$, modify $C$ accordingly before returning it to \adver, then simulate a ``blinded" oracle $D^t_\dk$ by decrypting $t'$-part queries as $k$.
	\item Finally, look at \adver's output. If \adver guesses $b'$ correctly, then \badver guesses $b= 0$, otherwise \badver guesses $b$ at random.
\end{enumerate}

The reduction works because if $b=0$, then we are simulating for \adver correctly either \game{2} or \game{3} (depending on $b'$). By assumption, in this case \adver should be able to guess correctly $b'$ with non-negligible advantage, otherwise his guess will be unrelated to $b$. Therefore, we have:

\begin{equation*}
	\Pr \left[ \adver \ \win\mbox{s} \ \game{3} \right] \geq \Pr \left[ \adver \ \win\mbox{s} \ \game{2} \right] - \negl.
\end{equation*}

Also notice that, at this point, \adver has no information whatsoever about the key $k$ used to encrypt the quantum part of the challenge, because the encryption of $k$ was blinded with another, unrelated one by the modified oracle $D^t_\dk$. This observation will be important in the next game hop, because the definition of \DNS security (and of the computational variant \cDNS) argues about adversaries averaged over the secret key used.

\Game{4:} as \game{3}, but this time \chall \ {\em always} replaces the challenge plaintext register $M$ with the entangled half $\phi^+_{M'}$, regardless of $b$. We show that this replacement cannot be efficiently detected. Clearly, if $b = 1$ nothing changes for \adver. If $b = 0$ instead, the quantum part of the challenge ciphertext received back by \adver is now an encryption (through $\Pi^\SKQE$) of the register $M'$ for a secret key $k$ unknown to \adver by the reasoning made in the last game. 

Then, consider any query to $D^t_\dk$ performed by $\adver_2$. Such query can be seen as a convex combination of: a query state which does not depend on the challenge (and hence does not change the \win probability), and the output of an attack map acting on the challenge ciphertext and an internal adversarial state $S$. By the \cDNS security of $\Pi^\SKQE$, such an attack map is (up to a negligible factor) a convex combination of two maps: the one which replaces the underlying plaintext with $\egoketbra{\bot}$ (and hence independent from the challenge plaintext, so again does not change the \win probability), and the one which acts as the identity, and leaves the underlying plaintext untouched. However, this latter map would produce a query which is detected as $\phi^+_{M'}$ by $D^t_\dk$, and would thus cause the adversary to \win. Hence we have:

\begin{equation*}
	\Pr \left[ \adver \ \win\mbox{s} \ \game{6} \right] \geq \Pr \left[ \adver \ \win\mbox{s} \ \game{5} \right] - \negl.
\end{equation*}

\Game{5:} as \game{4}, but we remove the ``$t$ blinding" feature from the simulated oracle $D^t_\dk$. That is, we replace $D^t_\dk$ back with $D'_\dk$. This replacement is undetectable for the same reasoning as in \game{3}.

\Game{6:} this is like \game{5}, except that:
\begin{enumerate}
	\item $D'_\dk$ is replaced by $D_\dk$ (i.e., aborts with \cheat instead of \win);
	\item \chall does not flip the random bit $b$; and
	\item on \adver's output, regardless of outcome, we declare \adver to \cheat or to \rej with 50/50 probability.
\end{enumerate}

Notice that this final game is equivalent to \QINDwCCAAfake. In fact, by now \adver's output is completely unrelated to $b$. Combining all inequalities, we have shown:
\begin{equation*}
	\Pr \left[ \adver \ \cheat\mbox{s in} \ \QINDwCCAAfake \right] \geq \Pr \left[ \adver \ \win\mbox{s} \ \QINDCCAAtest \right] - \negl ,
\end{equation*}

which proves finally that $\PiKEM$ is \QINDwCCAA secure.

\subsection{Strong (\cQCA-Based) CCA2 Security for Public-Key Quantum Encryption}\label{suppl:fullCCA2}

We show here the ``correct" (strongest) formulation of quantum CCA2 security for the public-key setting by using the framework from \expref{Section}{sec:PKEwCCA2} and extending in a natural way the work done in \cite{AGM18}. We recall that, in particular, this formulation has the advantage that it closely matches ``in spirit" the classical version of \INDCCAA by disallowing replay of the challenge ciphertext in a strict quantum sense, where detection of a replay is done by exploiting \QCA-based quantum ciphertext integrity. The drawback is that the use of such formulation in actual security proofs relies on a characterization of quantum encryption schemes given by a decomposition of the encryption/decryption procedure in terms of operators $V_k$ and $\Pi_{k, r}$, where $k$ (which in our case can be thought w.l.o.g. as the pair $(\dk,\ek)$) is the underlying \QCA authentication key (see \cite{AGM18} for details). One can prove that such decomposition always exists, but the necessary characterization here requires either that this decomposition is efficient (\expref{Condition}{con:efficient}), or that one accepts non-efficient challengers in the security game. There is currently no known counterexample of a quantum encryption scheme (either symmetric- or public-key) not respecting \expref{Condition}{con:efficient}. Regardless, as shown in \cite{AGM18}, this approach has many advantages, even if strictly speaking it can only be used on quantum encryption schemes which fulfill the characterization.

We start by recalling again the unrestricted quantum CCA2 experiment.

\begin{experiment}\label{exp:INDCCAA-test}
	The $\QINDCCAAtest(\Pi, \adver, \secpar)$ experiment:
	\begin{algorithmic}[1]
		\State $\chall$ runs $(\dk,\ek) \from \KeyGen(\secparam)$ and flips a coin $b \inrand \bit$;
		\State $\algo A_1$ receives $\ek$ and access to oracle $\Dec_\dk$;
		\State $\algo A_1$ prepares a side (state) register $S$, and sends $\chall$ a challenge register $M$;
		\State $\chall$ puts into $C$ either $\Enc_\ek(M)$ (if $b=0$) or $\Enc_\ek(\tau_M)$ (if $b=1$);
		\State $\algo A_2$ receives registers $C$ and $S$ and oracle access to $\Dec_\dk$;
		\State $\algo A_2$ outputs a bit $b'$.  \textbf{If} $b'=b$, \textbf{output} \textsf{win}; otherwise \textbf{output} \textsf{fail}.
	\end{algorithmic}
\end{experiment}

Notice that in this game there are no restrictions on the use of $\Dec_\dk$ by $\algo A_2$. In particular, $\algo A_2$ is free to decrypt the challenge. In the second game, the challenge plaintext is replaced by half of a maximally entangled state, and $\adver$ only gains an advantage over guessing if he cheats, i.e., if he tries to decrypt the challenge.

\begin{experiment}\label{exp:INDCCAA-fake}
	The $\QINDCCAAfake(\Pi, \algo A, n)$ experiment:
	\begin{algorithmic}[1]
		\State $\chall$ runs $k \from \KeyGen(1^n)$;
		\State $\algo A_1$ receives $\secpar$ and access to oracles $\Enc_k$ and $\Dec_k$;
		\State $\algo A_1$ prepares a side register $S$, and sends $\chall$ a challenge register $M$;
		\State $\chall$ discards $M$, prepares $\ket{\phi^+}_{M'M''}$ and fresh randomness $r$, and stores $(M'', r)$; then $\chall$ encrypts the $M'$ register and sends the resulting ciphertext $C'$ to $\algo A_2$;
		\State $\algo A_2$ receives registers $C'$ and $S$ and oracles $\Enc_k$ and $D_k$, where $D_k$ is defined as follows. On input a register $C$:
		\begin{enumerate}[(1)]
			\item $\chall$ applies $V_k^\dagger$ to $C$, places results in $MT$; 
			\item $\chall$ applies $\{P_T^{\sigma_k}, \one - P_T^{\sigma_k}\}$  to $T$; 
			\If{outcome is $0$}:
			\State $\chall$ applies $\{\Pi_{k, r}, \one - \Pi_{k, r}\}$  to $T$;  
			\If{outcome is $0$}: 
			\State $\chall$ applies $\{\Pi^+, \one - \Pi^+\}$ to $M M''$;
			\State \textbf{if} {outcome is $0$}: \textbf{output} \textsf{cheat};
			\EndIf
			\Else 
			\State apply the default map for invalid ciphertexts, i.e., $\hat D_k$ to $M$. 
			\EndIf
			\State	\textbf{return} $M$;
		\end{enumerate}
		\State $\chall$ draws a bit $b$ at random. \textbf{If} $b=1$, \textbf{output} \textsf{cheat}; if $b=0$ \textbf{output} \textsf{reject}.
	\end{algorithmic}
	
\end{experiment}

We now define \QINDCCAA in terms of the advantage gap of adversaries between the above two games.\footnote{The interface that the two games provide to the adversary differ slightly in that the adversary is not asked to output a bit in the end of the \QINDCCAAfake game. This is not a problem as the games have the same interface until the second one terminates.}

\begin{defn}\label{def:qINDCCAA}
	A $\PKQE$ $\Pi$ is $\QINDCCAA$ if, for all QPT adversaries $\adver$,
	$$
	\Pr[\QINDCCAAtest(\Pi, \adver, \secpar) \to \win] - \Pr[\QINDCCAAfake(\Pi, \adver, \secpar) \to \cheat] \leq \negl(n)\,.
	$$
\end{defn}

The omission of absolute values in the above is intentional. Indeed, an adversary can artificially inflate his cheating probability by querying the decryption oracle on the challenge and then ignoring the result. What he should not be able to do (against a secure scheme) is make his win probability larger than his cheating probability.

Separation and implication results follow the same strategy of \cite{AGM18}, which we refer the reader to. We only recall here the following.

\begin{prop}
	If $\Pi$ is a \QINDCCAA secure \PKQE scheme, then in particular it is also \QINDCCA and \QINDCPA.
\end{prop}

Moreover, following the same strategy as in \expref{Theorem}{thm:KEMwCCA2}, we have the following.

\begin{prop}\label{prop:KEMCCA2}
	Select the following:
	\begin{itemize}
		\item An \INDCCAA PKE scheme $\Pi^\PKE$;
		\item A \cQCA secure \SKQE $\Pi^\SKQE$. 
	\end{itemize}
	Then $\PiKEMP[\Pi^\PKE, \Pi^\SKQE]$ is a \QINDCCAA \PKQE.
\end{prop}

\subsection{Detailed Proof of \expref{Theorem}{thm:imp-2}}\label{app:imposs2}

We want to show that if $\Pi = (\KeyGen, \Sign, \Ver)$ is a correct \QS, and $M$ is a non-trivial two-outcome measurement, then $\Pi$ is at most $(1-\negl(\secpar))$-one-time $\{M\}$-secure. We use a similar idea as in the proof of \expref{Theorem}{thm:imp-1}. Just assuming correctness we can construct an attack that begins by applying a Stinespring dilation of the verification algorithm. Subsequently we apply a unitary that changes $\ket{\psi_0}$ to $\ket{\psi_1}$, where the $\ket{\psi_i}$ are efficiently preparable states such that measuring $M$ on $\ket{\psi_i}$ returns result $i$ with certainty. Finally, the attack undoes the Stinespring dilation of the verification. The result is a valid signed state for $\ket{\psi_1}$ when the attack is applied to a signed state for $\ket{\psi_0}$.

The details are as follows.
First note that \expref{Lemma}{lem:PKE-char} does not use the fact that the public key is public, or that the secret key is secret, it only uses the correctness of the scheme. Therefore we can apply it to the present \QS as well to conclude that the quantum channel implemented by the verification algorithm fulfills the equation
\begin{align}
	&\left\|\Sign_\sk- V_k\left((\cdot)\otimes (\sigma_k)_T\right)V_k^\dagger\right\|_\diamond\le\negl(\secpar)\label{eq:Signchar-sig}\\
	&\left\|\Ver_\vk( V_{k}P^{\sigma_k}_T V_k^\dagger(\cdot)V_kP^{\sigma_k}_T V_k^\dagger) - \tr_{T}\left[P^{\sigma_k}_T \left(V_k^\dagger (\cdot)V_k\right) P^{\sigma_k}_T \right] \right\|_\diamond\le\negl(\secpar),\label{eq:Signchar-ver}
\end{align}
where $k={\vk,\sk}$.\footnote{$V_k$ is not necessarily efficiently implementable, i.e., it does not make sense to say  it ``only depends on the verification key''; likewise for the state $\sigma_k$.}  In particular, \expref{Equation}{eq:Signchar-ver} implies that there exists a Stinespring dilation unitary $W^{\Ver_{\vk}}$ of $\Ver_{\vk}$ such that
\begin{align}\label{eq:sigchar}
	&\left\|W^{\Ver_{\vk}}V_{k}P^{\sigma_k}_T V_k^\dagger(\cdot)V_kP^{\sigma_k}_T V_k^\dagger\left(W^{\Ver_{\vk}}\right) - P^{\sigma_k}_T V_k^\dagger(\cdot)V_kP^{\sigma_k}_T \right\|_\diamond\le \negl(\secpar).
\end{align}
Let $\hat{W}^{\Ver_{\vk}}_{CE_1\to ME_2}$ be an efficiently implementable Stinespring dilation unitary of $\Ver_{\vk}$. We consider the following attack. Let $\ket{\psi_i}_{M}$, $i=0,1$ be efficiently preparable pure states such that $M$ returns $i$ with certainty when applied to $\proj{\psi_i}_M$, and let $U_i$ be preparation unitaries for $\ket{\psi_i}_M$, i.e. $\ket{\psi_i}_M=U_i\ket 0_M$. The attack map is now given by
\begin{align*}
	\adver(X)&=\tr_{E_1}\left[U_\adver XU_\adver^\dagger\right] \text{ with }\ 
	U_\adver=\left(\hat{W}^{\Ver_{\vk}}\right)^\dagger U_1U_0^\dagger\hat{W}^{\Ver_{\vk}}.
\end{align*}
Let us define $P:=V_{k}P^{\sigma_k}_T V_k^\dagger$. According to \expref{Equation}{eq:sigchar}, we have
\begin{align}\label{eq:adverchar}
	&\left\|P\adver\left(P(\cdot)P\right)P-P\tilde U_\adver \left(P(\cdot)P\right)\tilde U_\adver^\dagger P\right\|_\diamond\le\negl(\secpar)\text{ with }\ \tilde U_\adver=V_kU_1U_0^\dagger V_k^\dagger
\end{align}
When applying this attack to a ciphertext for $\proj{\psi_0}_{M}$ we calculate, using the symbol $\approx$ for equality up to negligible difference in trace norm,
\begin{align*}
	&P\adver\left(P\Sign_{\sk}\left(\proj{\psi_0}_M\right)P\right)P\approx P\tilde U_\adver \left(P\Sign_{\sk}\left(\proj{\psi_0}_M\right)P\right)\tilde U_\adver^\dagger P\\
	&\!\approx\! P\tilde U_\adver V_k\!\left(\!\proj{\psi_0}_M\!\otimes\! (\sigma_k)_T\!\right)\!V_k^\dagger\tilde U_\adver^\dagger P\!=\!PV_k\!\left(\!\left(\!U_1U_0^\dagger\proj{\psi_0}_MU_0U_1^\dagger\!\right)\!\!\otimes\! (\sigma_k)_T\!\right)\!V_k^\dagger P\\		
	&=V_k\left(\proj{\psi_1}_M\otimes (\sigma_k)_T\right)V_k^\dagger\approx \Sign_{\sk}\left(\proj{\psi_1}_M\right).
\end{align*}
Here we have used \expref{Equation}{eq:adverchar} in the second line,  and \expref{Equation}{eq:Signchar-sig} in the third and last lines. In particular, the calculation above implies that 
\begin{align*}
	&P\adver\left(P\Sign_{\sk}\left(\proj{\psi_0}_M\right)P\right)P\approx \adver\left(\Sign_{\sk}\left(\proj{\psi_0}_M\right)\right).
\end{align*}
We have thus transformed a signature of $\ket{\psi_0}$ into a signature of  $\ket{\psi_1}$ and thus changed the outcome of $M$ from 0 to 1, up to negligible difference.

\subsection{Detailed Proof of \expref{Theorem}{thm:one-time}}\label{app:one-time}

Let $\Pi^\SC$ be a classical signcryption scheme with one-time outsider and insider security, and let $\Pi^\SKQE$ be a \cDNS-secure symmetric-key quantum encryption scheme. We want to prove that $\PiHyb[\Pi^\SC, \Pi^\SKQE]$ (\expref{Construction}{cons:QS-generic}) is a one-time outsider- and insider-secure quantum signcryption scheme.

Insider security follows from the \QINDCPA security of the hybrid \PKQE construction (i.e. \expref{Theorem}{thm:PKE-hybrid-full}) and the fact that \cDNS implies \QIND. So let us focus on proving outsider security as from \expref{Definition}{def:one-time-outsider}.

Let $\adver = (\adver_1,\adver_2)$ be an outside attacker that prepares a quantum state $\rho_{MB}^{(\vek_R,\vek_S)}$ upon input $(\vek_R, \vek_S)$, and attacks the ciphertext it receives with an attack channel $\Lambda_{CB \to CB}$. First we will go through a series of indistinguishable hybrid experiments, starting from the outsider-security \expref{Experiment}{expr:one-time-outsider}. Then we will build a simulator for the last hybrid which actually implements an attack map of the desired convex form (identity plus rejection map), thereby proving security.

\Game{0:} this is just \expref{Experiment}{expr:one-time-outsider} for $\Pi = \PiHyb[\Pi^\SC, \Pi^\SKQE]$, that we recall here as follows.

\begin{experiment}
	The $\game{0}(\Pi, \adver, \secpar)$ experiment:
	\begin{algorithmic}[1]
		\State generate $(\sdk_S, \vek_S) \from \KeyGen^\SC(\secparam)$ and $(\sdk_R, \vek_R) \from \KeyGen^\SC(\secparam)$;
		\State $\adver_1$ receives $\vek_S$ and $\vek_R$;
		\State $\adver_1$ prepares a state $\rho_{MB}$;
		\State apply $\SigEnc^\Hyb_{\sdk_S, \vek_R}$ to register $M$, yielding register $C$;
		\State $\adver_2$ receives registers $C$ and $B$;
		\State $\adver_2$ applies an attack map $\Lambda_{CB \to CB}$ (possibly depending on $\vek_S, \vek_R$);
		\State apply $\VerDec^\Hyb_{\vek_S, \sdk_R}$ to register $C$, yielding register $M$; \textbf{output} $M$.
	\end{algorithmic}
\end{experiment}

\Game{1:} as \game{0}, but we remove the dependency of the attack map $\Lambda$ on $\vek_S, \vek_R$ by redefining $\rho_{MB} = \rho_{MB}^{(\vek_R,\vek_S)}$ to include a copy of the public keys in the $B$ register, and having $\Lambda$ read it from there. From now on, for ease of reading, we will omit this explicit labeling of the quantum state and just write $\rho_{MB}$ instead.

\Game{2:} as \game{1}, but we enforce that the attack map does not change the classical part of the ciphertext. Formally: let $C=C^\Qu C^\Cl$ be a decomposition of the ciphertext space of the scheme $\CQHyb$, where $C^\Qu$ and $ C^\Cl$ are a quantum and a classical register, respectively\footnote{These are both quantum registers, but we use this labeling for convenience, in order to identify the register which is supposed to contain the classical part of the ciphertext.}.

Let $\sigma_{CB} = \left( \SigEnc^\Hyb_{\sdk_S,\vek_R} \otimes \id_B \right) \rho_{MB}$ be the joint system that $\adver_2$ receives, and let $c$ be its classical part, i.e., the content of the $C^\Cl$ register. Let further $\hat\sigma_{CB}=\Lambda_{CB \to CB} \sigma_{CB}$ be the processed ciphertext after the attack map, and let $\hat c$ be the classical part of it. We observe that if $\hat c \neq c$, then $\VerDec^\Hyb_{\vek_S, \sdk_R}$ rejects with overwhelming probability due to the unforgeability of the classical signcryption scheme $\Pi^\SC$. We can therefore safely modify the experiment in such a way that whenever $\hat c \neq c$, it aborts by outputting $\egoketbra{\bot}_M$, and this modification (including the measurement on $C^\Cl$, that would happen anyway during \VerDec) is undetectable except with negligible probability.

\Game{3:} as \game{2}, but we ``blind" the classical part of the ciphertext, replacing it with a different one (a signcryption of a fresh symmetric key), and $\VerDec^\Hyb_{\vek_S, \sdk_R}$ is modified in order to undo the blinding when queried on the fake ciphertext, making the substitution transparent. The experiment so far looks as follows.

\begin{experiment}\label{exp:game3hyb}
	The $\game{3}(\Pi, \adver, \secpar)$ experiment:
	\begin{algorithmic}[1]
		\State generate $(\sdk_S, \vek_S) \from \KeyGen^\SC(\secparam)$ and $(\sdk_R, \vek_R) \from \KeyGen^\SC(\secparam)$;
		\State $\adver_1$ receives $\vek_S$ and $\vek_R$;
		\State $\adver_1$ prepares a state $\rho_{MB} = \rho_{MB}^{(\vek_R,\vek_S)}$;
		\State apply $\SigEnc^\Hyb_{\sdk_S, \vek_R}$ to register $M$, yielding register $C$;
		\State measure register $C^\Cl$, obtaining classical part of ciphertext $c$; record $c$;
		\State generate fresh $k' \leftarrow \KeyGen^\SKQE(\secparam)$
		\State encrypt $c' \from \SigEnc^\SC_{\sdk_S, \vek_R}(k')$; replace $c$ with $c'$ in $C^\Cl$;
		\State $\adver_2$ receives registers $C$ and $B$;
		\State $\adver_2$ applies an attack map $\Lambda_{CB \to CB}$;
		\State measure register $C^\Cl$, obtaining classical part of ciphertext $\hat c$;
		\State \textbf{if} $\hat c \neq c'$ \textbf{output} $\egoketbra{\bot}_M$; \textbf{else} replace $c'$ with $c$ in $C^\Cl$;
		\State apply $\VerDec^\Hyb_{\vek_S, \sdk_R}$ to register $C$, yielding register $M$; \textbf{output} $M$.
	\end{algorithmic}
\end{experiment}

Notice that the replacement of $c$ by $c'$ is undetectable by the adversary, because of the \IND security of $\Pi^\SC$. Therefore, the probability that the attack map $\Lambda$ produces a different $\hat c$ (and hence causes the game to abort) remains unchanged up to a negligible factor (otherwise it would yield an efficient \IND distinguisher). Thus, so far we have an experiment computationally equivalent to the original one. However, notice that at this point 
\adver's attack map $\Lambda_{CB \to CB} = \Lambda_{CB}$ can be decomposed as follows:

\begin{equation}\label{eq:decomp1}
	\Lambda_{CB} = \Lambda_{C^\Cl C^\Qu B} \approx \left( \id_{C^\Cl} \otimes \Gamma_{C^\Qu B} \right) + \Delta_{C^\Cl C^\Qu B} , 
\end{equation}

where the component $\Delta_{C^\Cl C^\Qu B}$ yields an effective attack map which always produces a rejection by the observation made during \game{2}, i.e.:

\begin{equation}\label{eq:decomp2}
	\tilde{\Delta}^{(\adver,S,R)}_{M B} = \mathbb{E} \left[ \VerDec^\Hyb_{S,R} \circ \Delta_{C^\Cl C^\Qu B} \circ \SigEnc^\Hyb_{S,R} \right] = \egoketbra{\bot}_{M} \otimes \Theta^\rej_{B} ,
\end{equation}

and where the expectation is over the keys produced by \KeyGen and the randomness of \SigEnc. Looking at the attack map component $\Gamma_{C^\Qu B}$, we notice that we can actually see this as an attack map of an adversary \badver against $\Pi^\SKQE$ (with quantum internal memory $B$, plaintext space $M$, and ciphertext space $C^\Qu$)\footnote{$\Gamma_{C^\Qu B}$ represents the adversary's attempt at forging a valid quantum ciphertext for the unknown secret key encapsulated in the classical part of the ciphertext.}. Let us consider the corresponding effective attack map:

\begin{equation}\label{eq:decomp3}
	\begin{split}
		\tilde{\Gamma}^{(\badver,S,R)}_{M B} &= \mathbb{E} \left[ \Dec^\SKQE \circ \Gamma_{C^\Qu B} \circ \Enc^\SKQE \right] \\
		&= \mathbb{E} \left[ 
		\VerDec^\Hyb_{S,R} \circ \left( \id_{C^\Cl} \otimes \Gamma_{C^\Qu B} \right) \circ \SigEnc^\Hyb_{S,R} 
		\right] .
	\end{split}
\end{equation}

Because of the \cDNS security of $\Pi^\SKQE$, we know that there exists a \cDNS simulator \simu implementing a computationally equivalent map:

\begin{equation}\label{eq:decomp4}
	\Phi_{M B} = \left( \id_M \otimes \Phi^\acc_{B} \right) + \left( \egoketbra{\bot}_M  \otimes \Phi^\rej_{B} \right) \approx \tilde{\Gamma}^{(\badver,S,R)}_{M B} .
\end{equation}

We can finally use \simu and the linearity of the maps defined so far to build an efficient outsider-security simulator for \adver. In fact, looking at \adver's effective attack map:

\begin{align*}
	\tilde{\Lambda}^{(\adver,S,R)}_{M B} &= \mathbb{E} \left[ \VerDec^\Hyb_{S,R} \circ \Lambda_{C^\Cl C^\Qu B} \circ \SigEnc^\Hyb_{S,R} \right] & \tag{from \eqref{eq:decomp1}} \\
	&\approx \mathbb{E} \left[ \VerDec^\Hyb_{S,R} \circ \left( \id_{C^\Cl} \otimes \Gamma_{C^\Qu B} \right) \circ \SigEnc^\Hyb_{S,R} \right. \\
	&\left. \ \ \ \ \ \ + \VerDec^\Hyb_{S,R} \circ \Delta_{C^\Cl C^\Qu B} \circ \SigEnc^\Hyb_{S,R} \right] \\
	&= \mathbb{E} \left[ \VerDec^\Hyb_{S,R} \circ \left( \id_{C^\Cl} \otimes \Gamma_{C^\Qu B} \right) \circ \SigEnc^\Hyb_{S,R} \right] \\
	&\ \ \ \ \ \ \ + \mathbb{E} \left[ \VerDec^\Hyb_{S,R} \circ \Delta_{C^\Cl C^\Qu B} \circ \SigEnc^\Hyb_{S,R} \right] & \tag{from \eqref{eq:decomp2} and \eqref{eq:decomp3}}\\
	&= \tilde{\Gamma}^{(\badver,S,R)}_{M B} + \egoketbra{\bot}_{M} \otimes \Theta^\rej_{B} & \tag{from \eqref{eq:decomp4}} \\
	&\approx \left( \id_M \otimes \Phi^\acc_{B} \right) + \left( \egoketbra{\bot}_M  \otimes \left( \Phi^\rej_{B} + \Theta^\rej_{B} \right) \right) \\
	&= \left( \id_M \otimes \Psi^\acc_{B} \right) + \left( \egoketbra{\bot}_M  \otimes \Psi^\rej_{B} \right) ,
\end{align*}

where $\Psi^\acc_{B} := \Phi^\acc_{B}$ and $\Psi^\rej_{B} := \Phi^\rej_{B} + \Theta^\rej_{B}$. This concludes the proof.

\subsection{Detailed Proof of~\expref{Theorem}{thm:many2hyb}}\label{sec:many2hyb}

In this section we 
prove outsider security of $\PiHyb\left[\Pi^{\SC}, \Pi^{\SKQE}\right]$ (\expref{Construction}{cons:QS-generic}) if $\Pi^{\SC}$ is an outsider-secure classical signcryption scheme, and $\Pi^{\SKQE}$ is a \cDNS-secure \SKQE. We will do this by showing that the (\cDNS-style version of) outsider security of a \QSC (resp., \SC) scheme is equivalent to the \wQAE (resp., \wCAE) security of a derived \SKQE (resp., \SKES) that is equivalent to the hybrid scheme \PiKEMS. Then, we will use \expref{Theorem}{thm:QAEhybrid} to conclude the proof. We first define a generic way to obtain a \SKQE from a \QSC, or a \SKES from a \SC, in a way that preserves secrecy and unforgeability. The intuition is easy: just define the secret key as the two signcryption key pairs of both users.

\begin{construction}
Let $\Pi=(\KeyGen, \SigEnc, \VerDec)$ be a \QSC (\SC). We define the \SKQE (\SKES)  $\Pi^\sharp=(\widehat\KeyGen, \Enc, \Dec)$ in the following way:
\begin{enumerate}
\item $\widehat\KeyGen$ runs $\KeyGen$ twice to obtain $k=(\sdk,\vek,\sdk',\vek')$;
\item $\Enc_k$  runs $\SigEnc_{\sdk, \vek'}$ and appends $\vek$ and $\vek'$, i.e.,\\
	$
	\Enc_k(X)=\left(\SigEnc_{\sdk, \vek'}, \vek,\vek'\right);
	$
\item $\Dec_k$ checks whether the second and third part of a ciphertext are equal to $ \vek$ and $\vek'$. If one of them is not, it outputs $\bot$, otherwise it runs $\VerDec_{\sdk', \vek}$ on the first part of the ciphertext.
\end{enumerate}
\end{construction}

We can prove a direct connection between the security of the underlying \QSC (resp., \SC) scheme, and the security of the resulting \SKQE (resp., \SKE) scheme.

\begin{lem}\label{lem:SC2AE}
	Let $\Pi=(\KeyGen, \SigEnc, \VerDec)$ be a \QSC (resp., \SC). Then $\Pi^\sharp$ is \wQAE (resp., \CAE) if and only if $\Pi$ is many-time, outsider secure (resp., many-time outsider secure in the strong unforgeability sense).
\end{lem}
\begin{proof}
	Let us first look at the classical case. If $\Pi$ is an outsider secure \SC (in the strong unforgeability sense), then it is by definition strongly unforgeable and \INDCCAA secure. 
Now note that these properties carry over to $\Pi^\sharp$ using the following argument. Given a forging or \INDCCAA adversary \adver against $\Pi$, we can build an adversary \adver' against $\Pi^\sharp$ by making one arbitrary extra query at the beginning to obtain the public keys, and then running \adver, ignoring the second and third parts of all ciphertexts. Whenever \adver submits a ciphertext to the verified decryption oracle or as a forgery, \adver' appends the public keys before submission. \adver' has the same output distribution as \adver.
On the other hand, let \adver now be an adversary against $\Pi^\sharp$. We construct an adversary \adver' against $\Pi$ in the following way. \adver' runs \adver, relaying any encryption queries to the signcryption oracle and appending the public keys to the result. The same is done in the challenge phase of the \INDCCAA game, here the challenge plaintext submitted by \adver is forwarded as the challenge  in the \INDCCAA game for $\Pi$, and when the challenge ciphertext is received, the public keys are appended. For decryption queries, or forgeries, \adver' checks whether the second and third parts of the ciphertext are equal to the public keys before relaying it to the decryption oracle, or submitting it as a forgery, respectively. The observation that \CAE is equivalent to integrity of ciphertext and \INDCCAA finishes the proof in the classical case.
	
	Let us now turn to the quantum case. The two adversary transformations above simulate the \OutReal and \OutIdeal games for $\Pi$ using oracles provided by the \QAEreal and \wQAEideal games for $\Pi^\sharp$, respectively, and vice versa.
\end{proof}

As a corollary, we obtain outsider security of the hybrid construction with an outsider-secure \SC and a \cDNS secure \SKQE.

\begin{cor}\label{cor:outsider}
	Let $\Pi^{\Cl}$ be a many-time, outsider secure \SC and let $\Pi^\Qu$ be a \cDNS-secure \SKQE. Then $\PiHyb[\Pi^{\Cl},\Pi^\Qu]$ is a many-time outsider-secure \QSC. 
\end{cor}
\begin{proof}
	This follows immediately from \expref{Theorem}{thm:QAEhybrid} and \expref{Lemma}{lem:SC2AE}, together with the observation that
$
\left(\PiHyb[\Pi^{\Cl},\Pi^\Qu]\right)^\sharp=\PiHyb\left[\left(\Pi^{\Cl}\right)^\sharp,\Pi^\Qu\right].
$

\end{proof}

\subsection{Explanatory Example for \expref{Condition}{con:efficient}}

\expref{Condition}{con:efficient} is weaker than the analogous condition from \cite{AGM18} where it is required that $V_k$ is efficiently implementable \emph{on all inputs.} The results of that paper also hold for all schemes satisfying this weaker \expref{Condition}{con:efficient}. As discussed in \cite{AGM18}, all known \SKQE schemes satisfy \expref{Condition}{con:efficient}.

For example, let's look at why $\PiKEMP$ satisfies \expref{Condition}{con:efficient}, provided that the same holds for the underlying \SKQE. Let $\Pi_Q$ be the \SKQE and $\Pi^\Cl$ the classical \PKE. Suppose $\Pi^\Qu$ satisfies
$$
\Enc_{k;r}(X_M) =V_k \left(X_M \otimes \proj{\psi^{(k,r)}}_T \right) V_k^{\dagger}.
$$
for some efficiently-preparable $\ket{\psi^{(k, r)}}$ and an $V_k$. Consider the following way of implementing encryption of the combined scheme $\PiKEMS[\Pi^\Qu, \Pi^\Cl]$. Here $r_C$, $k$, $r_Q$ denote uniformly random classical strings which are (respectively) the randomness for $\Pi^\Cl$ encryption, a key for $\Pi^\Qu$, and randomness for $\Pi^\Qu$.
\begin{align*}
\rho 
&\longmapsto \rho \otimes \proj{k}_K \otimes \proj{\Enc^C_{\ek; r_C}(k)} \otimes \proj{\psi_{k, r_Q}}\\
&\longmapsto \tr_K \left(\proj{k}_K \otimes \proj{\Enc^C_{\ek; r_C}(k)} \otimes V_k ( \rho \otimes \proj{\psi_{k, r_Q}})V_k^\dagger\right)
\end{align*}
Here, the first step attaches the auxiliary state to $\rho$, and the second step implements a unitary $V^\Pi_k$ (on relevant inputs.) We emphasize that both steps are efficiently implementable. Note that, despite the fact that (for these particular inputs) we implemented $V^\Pi_k$ by tracing out a register, it is still the case that $V^\Pi_k$ is a unitary operator. Indeed, $V^\Pi_k$ is invertible because $k$ can be extracted from $\Enc^C_{\ek; r_C}(k)$; of course, to do this efficiently one needs $\dk$.

\subsection{Another Construction for Quantum Signcryption}\label{sec:QSCmod}

Here we propose (without arguing about security) an idea for a different construction for a \QSC scheme, based on a classical \PKE and a classical \DS scheme rather than a classical \SC. The idea is to \PKE-encrypt a symmetric quantum encryption key, and at the same time embedding a signature of the symmetric key inside the ciphertext, instead of the $0,\ldots,0$ authentication tag. The resulting scheme has the advantage that in certain cases it might result in more compact ciphertexts. We leave the formal cryptographic security study of this scheme as an open problem for further investigation.

\begin{construction}\label{cons:QS-modular}
Let $\Pi^\PKE$ be a classical \PKE scheme, $\Pi^\DS$ a classical \DS, and $\Pi^\SKQE$ a \SKQE
. Define a quantum signcryption scheme $\PiMod[\Pi^\PKE, \Pi^\DS, \Pi^\SKQE] = (\KeyGen, \SigEnc, \VerDec)$ as follows:
\begin{enumerate}
	\item $\KeyGen(\secparam):$ output $(\sdk, \vek)$, where $\sdk = (\sk,\dk), \vek = (\vk,\ek),\\  (\sk,\vk) \from \KeyGen^{\DS}(\secparam)$, and $(\dk,\ek) \from \KeyGen^{\PKE}(\secparam)$;
	\item $\SigEnc_{\sdk,\vek}:$ on input $\rho_M$, do the following:
		\begin{enumerate}
		\item generate $k \from \KeyGen^{\SKQE}(\secparam)$
		\item encrypt $t \from \Enc^{\PKE}_\ek(k)$
		\item sign $s \from \Sign^{\DS}_\sk(k)$
		\item generate state $\psi_C \from \Enc^{\SKQE}_k(\rho_M \otimes \egoketbra{s})$
		\item output $(\psi_C,t)$
		\end{enumerate}
	\item $\VerDec_{\sdk,\vek}:$ on input $(\psi_C,t)$, do the following:
		\begin{enumerate}
		\item decrypt $k \from \Dec^{\PKE}_\dk(t)$; if decryption is $\bot$ then output $\egoketbra{\bot}$
		\item decrypt $\varphi_C \from \Dec^{\SKQE}_k(\psi_C)$; if decryption is $\egoketbra{\bot}$ then output $\egoketbra{\bot}$
		\item perform a partial measurement on $\varphi_C$, obtaining classical outcome $s$ and partially traced out state $\rho_M$
		\item verify that $\Ver{\DS}_\vk(k,s) \to \acc$: if not, then output $\egoketbra{\bot}$, otherwise output $\rho_M$.
		\end{enumerate}
\end{enumerate}
\end{construction}

\end{document}